\documentclass[onecolumn,draftcls,11pt]{IEEEtran}

\usepackage[utf8]{inputenc}
\usepackage[T1]{fontenc}
\usepackage{url}
\usepackage{ifthen}
\usepackage{cite}
\usepackage[cmex10]{amsmath}
\usepackage{amssymb}
\usepackage{amsthm}
\usepackage{mathtools}
\usepackage{graphicx}
\usepackage{bbm}
\usepackage{xcolor}
\usepackage{tikz}
\usepackage{lmodern}
\usepackage{todonotes}
\usepackage[most]{tcolorbox}
\usepackage{algorithm}
\usepackage{algpseudocode}

\usetikzlibrary{
    patterns,
    calc,
    intersections,
    shapes,
    shapes.misc,
    arrows.meta
}

\usepackage[
    colorlinks=true,
    linkcolor=blue,
    citecolor=blue,
    urlcolor=blue
]{hyperref}

\newtheorem{theorem}{Theorem}
\newtheorem{lemma}{Lemma}
\newtheorem{corollary}{Corollary}

\theoremstyle{definition}

\newcommand{\Nnodes}{n}

\newcommand{\ndim}{d}

\newcommand{\Nhonest}{h}

\newcommand{\Nadversarial}{t}

\newcommand{\allnodesset}{\mathcal{V}}

\newcommand{\honestset}{\mathcal{K}}

\newcommand{\adversaryset}{\mathcal{Q}}

\newcommand{\inputset}{\mathcal{U}}

\newcommand{\bU}{\mathbf{U}}
\newcommand{\bu}{\mathbf{u}}

\newcommand{\bY}{\mathbf{Y}}
\newcommand{\by}{\mathbf{y}}

\newcommand{\bN}{\mathbf{N}}
\newcommand{\bn}{\mathbf{n}}

\newcommand{\bx}{\mathbf{x}}
\newcommand{\bz}{\mathbf{z}}

\newcommand{\bL}{\mathbf{L}}
\newcommand{\bR}{\mathbf{R}}

\newcommand{\mE}{\mathbb{E}}
\newcommand{\DC}{\mathsf{DC}}
\newcommand{\AD}{\mathsf{AD}}
\newcommand{\sU}{\mathsf{U}}
\newcommand{\est}{\mathsf{est}}
\newcommand{\mse}{\mathsf{MSE}}
\newcommand{\pa}{\mathsf{PA}}

\newcommand{\modeldep}{\mathrm{dep}}
\newcommand{\modelind}{\mathrm{ind}}

\newcommand{\acceptanceevent}[1]{\mathcal{A}_{#1}}

\newcommand{\gameinstance}[1]{%
    \mathfrak{G}_{\ndim,\Nnodes,\Nhonest,\Nadversarial}^{#1}
}

\newcommand{\performanceregion}[1]{%
    \mathcal{R}_{\ndim,\Nnodes,\Nhonest,\Nadversarial}^{#1}
}

\DeclareMathOperator{\supp}{supp}
\DeclareMathOperator*{\argmax}{arg\,max}
\DeclareMathOperator*{\argmin}{arg\,min}

\begin{document}

\title{
    Game of Coding under Computation-Dependent Adversarial Noise
    \thanks{
        The work of Mohammad Ali Maddah-Ali and Hanzaleh Akbari
        Nodehi has been partially supported by the National Science
        Foundation under Grant CCF-2348638.
    }
}

\author{
    Hanzaleh Akbari Nodehi, Mohammad Ali Maddah-Ali\\
    University of Minnesota, Twin Cities, Minneapolis, MN, USA\\
    Email: \{akbar066,maddah\}@umn.edu
}
\maketitle

\begin{abstract}

The game of coding framework was introduced to extend coding-theoretic
recovery beyond its traditional limit, under which the number of
honest reports must exceed the number of adversarial or corrupted
reports. It does so by exploiting the rational behavior of adversarial
participants and their incentive to keep the system live. Existing
game-of-coding formulations, however, assume that the
adversarial-noise distribution is independent of the realized
ground-truth computation. This assumption may be restrictive when an
informed adversary can adapt its reports to the value being computed.

In this paper, we study the game of coding with input-dependent
adversarial noise. We introduce a unified multi-node,
multidimensional formulation. For every family of conditional adversarial-noise
distributions, we construct an input-independent joint noise
distribution, and prove that this reduction exactly preserves the probability of
acceptance and the accepted mean-squared estimation error. Consequently, the input-dependent and
input-independent models have identical achievable performance
regions,  and the same
equilibrium utilities. 

\end{abstract}

\section{Introduction}\label{sec:introduction}

Coding theory uses redundancy to protect information against
unreliable or adversarial components and is fundamental to
communication, storage, and computing systems
\cite{SudanBook,ZamirCoded,yu2017polynomial,yu2019lagrange}. It has also been extended from
discrete to analog domains, enabling approximate recovery \cite{roth2020analog,jahani2018codedsketch,
BACC}. Despite
these advances, both discrete and analog coding schemes rely on a
fundamental trust requirement. To illustrate, consider a coded
computing system with node set
$\allnodesset=\honestset\cup\adversaryset$, where nodes in
$\honestset$ follow the prescribed protocol, while nodes in
$\adversaryset$ may deviate arbitrarily. Repetition coding requires
\(
    |\honestset|\geq|\adversaryset|+1
\)
for error-free recovery. Similarly, an
$(K,\Nnodes)$ Reed--Solomon code \cite{SudanBook} requires
\(
    |\honestset|\geq|\adversaryset|+K
\),
while Lagrange coded computing \cite{yu2019lagrange} for a polynomial of degree $d$ requires
\(
    |\honestset|>|\adversaryset|+(K-1)d
\). Thus, from repetition coding to more
advanced discrete and analog coding schemes, reliable recovery
requires the honest nodes to outnumber the adversarial nodes by a
scheme-dependent margin.

This limitation is particularly relevant to decentralized systems in
which computationally intensive tasks are outsourced to external
workers. Examples include decentralized machine learning, oracle
networks, and blockchain-based computation, where participants may be
anonymous, dynamically selected, and economically motivated
\cite{bitcoin2008bitcoin,buterin2013ethereum,sliwinski2019blockchains,
han2021fact,gans2023zero}. In such environments, the data collector
(DC) may not be able to certify an honest majority before assigning a
task. Cryptographic verification can provide strong correctness
guarantees, but proof generation may impose substantial computational
and latency costs and may not naturally accommodate approximate
computations
\cite{thaler2022proofs,feng2021zen,liu2021zkcnn,xing2023zero,
mohassel2017secureml,lee2024vcnn,weng2021mystique,
chen2022interactive,garg2022succinct,setty2012taking,
garg2023experimenting}. Coding-based redundancy is comparatively
lightweight and naturally supports approximate computation, but its
traditional guarantees deteriorate when adversarial workers are not a
minority.

A key feature of many decentralized systems is that participants are
not merely malicious but economically rational. A worker is commonly
rewarded only when its submitted computation is accepted and the
system remains functional. An adversarial worker therefore faces two
competing incentives. It seeks to increase the error of the accepted
estimate, but an excessively inconsistent report may be rejected,
preventing the adversary from receiving a reward or influencing the
output. The DC faces the opposite accuracy objective while also
preferring the computation to remain live. This creates a strategic
interaction between accuracy and liveness that is not captured by a
pure worst-case adversarial model.

The \emph{game of coding} framework exploits the rationality of
adversarial participants by modeling the interaction between the DC
and the adversary as a Stackelberg game
\cite{GoCJournal,GoDSybil,nodehi2026game,nodehi2026mathsf,
nodehi2025unknown,akbari2026learning,nodehi2026gameBits}. The DC,
acting as the leader, commits to an acceptance rule, while the
adversary observes this rule and selects its reporting strategy. The
utilities of the players depend on two operational quantities: the
probability of acceptance, which captures system liveness, and the
estimation error conditioned on acceptance, which captures accuracy.
Since an adversary can receive a reward or influence the output only
when its reports are accepted, it must balance the desire to increase
the estimation error against the need to keep the system live.

The foundational work in \cite{GoCJournal} introduced the theoretical
framework for the scalar two-node setting with one honest node and one
adversarial node. It formalized the acceptance--estimation mechanism,
defined the utilities of the DC and the adversary, characterized their
Stackelberg equilibrium, and derived the corresponding optimal
strategies. Importantly, it showed that the system can maintain
positive liveness and provide a useful estimate even when one of the
two nodes is adversarial. This is a regime in which classical
repetition coding cannot guarantee reliable decoding because there is
no honest majority.

The framework was subsequently extended to the scalar multi-node
setting in \cite{GoDSybil}, with particular emphasis on
\emph{Sybil resistance}. This property is essential in open
decentralized systems, where an adversary may create many identities
and participate through multiple apparently distinct nodes. The
analysis showed that increasing the number of adversarial identities
does not improve the adversary's equilibrium utility or reduce the
DC's equilibrium utility. In this sense, the adversary gains no
additional strategic power by creating Sybil nodes.

The scalar restriction was removed in \cite{nodehi2026game}, which
extended the game of coding to vector-valued computations. That work
considered a general multidimensional computation output and
characterized the equilibrium acceptance rule, adversarial strategy,
and estimation performance in the vector setting. This extension is
particularly relevant to applications in which the outsourced result
is naturally a vector, including gradients, model updates, and
intermediate representations in machine learning.

The works above primarily considered a single outsourced computation.
The VISTA framework in \cite{nodehi2026mathsf} extended the game of
coding to multi-round iterative optimization, such as training a
machine-learning model. In this setting, the DC must account not only
for the accuracy and liveness of one round, but also for how each
accepted or rejected update affects the entire optimization
trajectory. VISTA jointly adapts the acceptance threshold and the
learning rate to obtain reliable long-term optimization performance
in adversary-dominated environments.

The preceding formulations assume that the utility of the adversary
is known to the DC. This complete-information assumption was relaxed
in \cite{nodehi2025unknown}. There, the DC learns the adversary's
response from repeated interactions and selects an acceptance
threshold whose resulting utility is close to that of the optimal
threshold available under complete information. The primary objective
is therefore the quality of the final threshold recommended after the
learning process.

A complementary learning objective was studied in
\cite{akbari2026learning}. Rather than evaluating only the final
threshold, that work accounts for the utility lost throughout the
learning trajectory. The problem is formulated through cumulative
regret, thereby penalizing poor threshold choices made while the DC is
still learning the adversary's behavior. The proposed learning policy
is designed to control this cumulative cost while progressively
identifying near-optimal acceptance thresholds.

Finally, \cite{nodehi2026gameBits} broadened the coding component of
the framework beyond repetition coding. It studied how the
game-theoretic principles underlying the game of coding can be
combined with more advanced coding constructions, including linear
codes. This direction shows that the framework is not tied to
duplicating a single computation, but can also exploit the structure
and efficiency of general coded computations.

Despite these developments, the existing game-of-coding formulations
share a common assumption: the adversarial-noise distribution is
independent of the realized ground-truth computation. Under the
additive representation $\bY_q=\bU+\bN_q$, the adversary selects one
noise distribution for $\bN_q$ and uses it for every realization of
$\bU$. The strategy may depend on the acceptance rule, the system
parameters, and the players' utilities, but it does not adapt to the
particular value being computed.

This assumption may be restrictive in practice. For example,
adversarial-example attacks construct perturbations as functions of
the particular input, the current model, and, in targeted attacks, the
desired output
\cite{szegedy2014intriguing,goodfellow2015explaining,
madry2018towards,carlini2017towards}. Similarly, model-poisoning
attacks in distributed and federated learning may adapt malicious
updates to the current global model, selected targets, available
information about benign updates, and the aggregation rule
\cite{bhagoji2019analyzing,baruch2019little,
bagdasaryan2020backdoor,fang2020local}. In both cases, the adversarial
action is tailored to the realized input or the current state of the
computation rather than being sampled from one fixed distribution.

The same issue arises naturally in the game of coding. If $\bU$
represents an inference output, a gradient, a model update, or another
state-dependent computation, an informed adversary may adapt the
direction, magnitude, and correlation structure of its reports to the
realized value of $\bU$. We model such behavior through a family of
conditional joint noise PDFs
$\{g_{\bu}\}_{\bu\in\inputset}$, where $g_{\bu}$ is the adversarial
noise distribution used when $\bU=\bu$. Since the input-independent
strategy set is contained in this more general strategy set, it is not
immediate whether the earlier game-of-coding guarantees remain valid.

\begin{tcolorbox}[
    colback=gray!10,
    colframe=gray!60!black,
    boxrule=0.6pt,
    arc=1mm,
    left=2mm,
    right=2mm,
    top=1.5mm,
    bottom=1.5mm
]
\centering
\textbf{Central question.}
Does allowing the adversarial-noise distribution to depend
arbitrarily on the realized ground-truth computation enlarge the
adversary's achievable performance region or change the resulting
Stackelberg equilibrium?
\end{tcolorbox}

Answering this question across the existing game-of-coding models
requires a common formulation. The previous works consider different
numbers of nodes, computation dimensions, and acceptance and estimation
rules. We therefore introduce a unified $\Nnodes$-node,
$\ndim$-dimensional model. 
As shown in Appendix~\ref{app:alignment_previous_goc}, this mechanism
reduces exactly to the scalar max--min and midrange mechanism when
$\ndim=1$, and to the two-node vector distance and midpoint mechanism
when $\Nnodes=2$.

Our main result shows that the additional dependence on $\bU$ provides
no performance advantage to the adversary. Given an arbitrary
conditional strategy $\{g_{\bu}\}_{\bu\in\inputset}$, we construct one
input-independent joint noise PDF by averaging $g_{\bu}$ over the
actual distribution of $\bU$. This construction applies to an
arbitrary distribution of $\bU$ and permits arbitrary statistical
dependence among the adversarial nodes.
We show that every input-dependent strategy has an
input-independent counterpart with exactly the same probability of
acceptance and accepted mean-squared estimation error. Conversely,
every input-independent strategy is already an admissible
input-dependent strategy. Therefore, the two models have identical
achievable performance regions for every acceptance policy.

\subsection{Main Contributions}

The main contributions of this paper are summarized as follows.

\begin{enumerate}
    \item \textbf{A unified game-of-coding formulation.}
    We formulate an $\Nnodes$-node, $\ndim$-dimensional game of coding
    with multiple honest and adversarial nodes, an arbitrary
    distribution for $\bU$, a joint input-dependent adversarial-noise. The model allows
    adversary-dominated regimes and arbitrary dependence among the
    adversarial reports. Appendix~\ref{app:alignment_previous_goc}
    shows that the scalar multi-node and two-node vector-valued
    mechanisms are exact special cases.

    \item \textbf{An exact reduction to an input-independent
    adversary.}
    For every conditional adversarial strategy
    $\boldsymbol{g}=\{g_{\bu}\}_{\bu\in\inputset}$, we construct an
    averaged input-independent joint noise PDF. We prove that the
    resulting strategy has exactly the same probability of acceptance
    and accepted mean-squared estimation error as the original
    strategy for every acceptance policy selected by the DC.

    \item \textbf{Equality of the achievable performance regions.}
    We prove that the input-dependent and input-independent models
    induce the same set of achievable probability-of-acceptance and
    mean-squared-error pairs for every threshold. Thus, allowing the
    adversary to adapt its noise distribution to the realization of
    $\bU$ does not enlarge its achievable performance region.

    \item \textbf{Equivalence of the Stackelberg games.}
    We show that the two models induce the same adversarial
    best-response performance and the same worst-case objective for
    the DC at every threshold. Consequently, they have the same set of
    optimal thresholds and the same equilibrium utilities.
\end{enumerate}

The remainder of the paper is organized as follows.
Section~\ref{sec:problem_formulation} presents the unified
input-dependent and input-independent game formulations.
Section~\ref{sec:main_results} states the reduction,
performance-region equivalence, and Stackelberg-equivalence results.
Appendix~\ref{app:alignment_previous_goc} establishes the relationship
with the previous scalar and vector-valued game-of-coding models. The
remaining appendices contain the proofs of the main results.

\subsection{Notation}\label{sec:notation}

We  denote random variables using uppercase letters and
deterministic values, using lowercase letters. We
distinguish vectors from scalars by using boldface type for vectors and
standard type for scalars. For example, $\bU$ represents a random
vector, whereas $\bu$ denotes one of its deterministic realizations.
Similarly, $U$ represents a scalar random variable, whereas $u$
denotes a deterministic scalar.

For a positive integer $m$, we use
$[m]\triangleq\{1,\ldots,m\}$.
For a finite set $\mathcal{S}$, its cardinality is denoted by
$|\mathcal{S}|$. The notation $\mathbb{R}^{\ndim}$ denotes the
$\ndim$-dimensional Euclidean space, and $\|\cdot\|_2$ denotes the
Euclidean norm.
For a random vector $\mathbf{X}$ and one of its realizations
$\mathbf{x}$, we denote its probability density function by
$f_{\mathbf{X}}(\mathbf{x})$. Similarly,
$f_{\mathbf{X}\mid\mathbf{Z}}(\mathbf{x}\mid\mathbf{z})$ denotes the
conditional PDF of $\mathbf{X}$ given $\mathbf{Z}=\mathbf{z}$. The
support of $f_{\mathbf{X}}$ is denoted by
$\supp(f_{\mathbf{X}})$.

For any set $\mathcal{S}\subseteq\mathbb{R}^{\ndim}$ and any function
$f:\mathcal{S}\to\mathbb{R}$, the notation $\argmax_{x\in\mathcal{S}} f(x)$
represents the set of all elements $x\in\mathcal{S}$ that maximize
$f(x)$. Similarly, $\argmin_{x\in\mathcal{S}} f(x)$
represents the set of all elements that minimize $f(x)$. For $a,b\in\mathbb{R}$ with $a\leq b$, the notation $[a,b]$
represents the closed interval $[a,b]\triangleq
    \left\{
        x\in\mathbb{R}:a\leq x\leq b
    \right\}$.


\section{Problem Formulation}\label{sec:problem_formulation}

We consider a system consisting of a data collector (DC) and
$\Nnodes\geq 2$ computational nodes. Let us define the set of all
computational nodes as $\allnodesset\triangleq[\Nnodes]$. The DC aims
to estimate an $\ndim$-dimensional random vector
$\bU\in\mathbb{R}^{\ndim}$, referred to as the ground-truth
computation. We assume that $\bU$ has an arbitrary probability
distribution with support $\inputset$. The DC does not have direct
access to the realization of $\bU$ and must estimate it using the
reports received from the computational nodes.

We partition $\allnodesset$ into a set of honest nodes $\honestset$
and a set of adversarial nodes $\adversaryset$. We define their
cardinalities as $|\honestset|=\Nhonest$ and
$|\adversaryset|=\Nadversarial$, respectively. These sets satisfy
\begin{align}
    \honestset\cap\adversaryset=\emptyset,\qquad
    \honestset\cup\adversaryset=\allnodesset,\qquad
    \Nhonest+\Nadversarial=\Nnodes.
    \label{eq:node_partition}
\end{align}
We assume that $\Nhonest\geq1$ and $\Nadversarial\geq1$, but impose no
honest-majority condition. In particular, our formulation permits
$\Nadversarial\geq\Nhonest$ and therefore includes
adversary-dominated regimes. This feature distinguishes the game of
coding formulation from classical adversarial error-correction
settings, in which reliable decoding generally requires sufficient
honest redundancy relative to the number of adversarial nodes
\cite{SudanBook}. The identities of the adversarial nodes are unknown
to the DC.

Each node $i\in\allnodesset$ sends a report
$\bY_i\in\mathbb{R}^{\ndim}$ to the DC. An honest node
$k\in\honestset$ and an adversarial node
$q\in\adversaryset$ report\footnote{The additive representation in
\eqref{eq:adversarial_report} imposes no restriction on the
adversary's reporting strategy. Indeed, any possibly randomized report
$\bY_q$ can be written as $\bY_q=\bU+\bN_q$ by defining
$\bN_q\triangleq\bY_q-\bU$. Since the conditional distribution of
$\bN_q$ is allowed to depend arbitrarily on the realization of $\bU$,
as formalized in \eqref{eq:conditional_adversarial_pdf}, this model
captures arbitrary adversarial reports.}, respectively,
\begin{align}
    \bY_k
    &=
    \bU+\bN_k,
    \qquad k\in\honestset,
    \label{eq:honest_report}\\
    \bY_q
    &=
    \bU+\bN_q,
    \qquad q\in\adversaryset.
    \label{eq:adversarial_report}
\end{align}
The honest noises $\{\bN_k\}_{k\in\honestset}$ are independent and
identically distributed according to a publicly known PDF
$f_{\bN}(\cdot)$. We assume that their magnitudes are bounded according
to
\begin{align}
    \Pr\!\left(\|\bN_k\|_2>\Delta\right)=0,
    \qquad k\in\honestset,
    \label{eq:honest_noise_bound}
\end{align}
where $\Delta>0$ is publicly known. The honest noise models
inaccuracies caused by approximate computation, randomized
algorithms, quantization, compression, measurement error, or oracle
inaccuracy.

Let us define the joint honest- and adversarial-noise vectors as
$\underline{\bN}_{\honestset}\triangleq
(\bN_k:k\in\honestset)$ and
$\underline{\bN}_{\adversaryset}\triangleq
(\bN_q:q\in\adversaryset)$, respectively. The adversary has access to
the exact realization of $\bU$ and may coordinate the noises of all
adversarial nodes. For every realization $\bu\in\inputset$, let
$\underline{\bz}_{\adversaryset}\triangleq
(\bz_q:q\in\adversaryset)$ denote a possible realization of
$\underline{\bN}_{\adversaryset}$. We use
$f_{\underline{\bN}_{\adversaryset}\mid\bU}
(\underline{\bz}_{\adversaryset}\mid\bu)$ to denote the conditional
joint PDF of the adversarial-noise vector
$\underline{\bN}_{\adversaryset}$ given $\bU=\bu$, and define
\begin{align}
    g_{\bu}\!\left(\underline{\bz}_{\adversaryset}\right)
    \triangleq
    f_{\underline{\bN}_{\adversaryset}\mid\bU}
    \!\left(
        \underline{\bz}_{\adversaryset}\mid\bu
    \right).
    \label{eq:conditional_adversarial_pdf}
\end{align}
Thus, for each realization $\bU=\bu$, the adversary may select a
different joint noise PDF $g_{\bu}(\cdot)$. Therefore, the distribution
of the adversarial noises may depend arbitrarily on $\bU$, and the
noises of the adversarial nodes may also be arbitrarily correlated with
one another.

We assume that the honest-noise vector is independent of $\bU$ and,
conditional on $\bU=\bu$, is independent of the adversarial-noise
vector. Equivalently, for every $\bu\in\inputset$, the conditional
joint PDF satisfies
\begin{align}
&f_{\underline{\bN}_{\honestset},
\underline{\bN}_{\adversaryset}\mid\bU}
\left(
    \underline{\bx}_{\honestset},
    \underline{\bz}_{\adversaryset}
    \mid\bu
\right)
\nonumber\\
&\qquad =
f_{\underline{\bN}_{\honestset}}
\left(
    \underline{\bx}_{\honestset}
\right)
g_{\bu}
\left(
    \underline{\bz}_{\adversaryset}
\right).
\label{eq:conditional_honest_adversarial_factorization}
\end{align}
Hence, the adversary may adapt its noise distribution to the
realization of $\bU$, but it cannot observe or statistically
couple its reports with the realized honest noises.

We define an input-dependent adversarial strategy as the collection
$\boldsymbol{g}\triangleq\{g_{\bu}\}_{\bu\in\inputset}$. The
corresponding adversarial action set is
\begin{align}
\Lambda_{\AD}^{\modeldep}(\Nadversarial)
\triangleq
\Big\{
    \boldsymbol{g}=\{g_{\bu}\}_{\bu\in\inputset}
    \,\Big|\,
    g_{\bu}
    \text{ is a valid joint PDF on }
    \big(\mathbb{R}^{\ndim}\big)^{\Nadversarial},
    \ \forall\bu\in\inputset
\Big\}.
\label{eq:dependent_adversarial_action_set}
\end{align}
As in the previous game-of-coding formulations
\cite{GoCJournal,GoDSybil,nodehi2026game,nodehi2026mathsf,
nodehi2025unknown,akbari2026learning,nodehi2026gameBits}, this action set also includes discrete
and mixed strategies represented using Dirac delta functions. For
comparison, we define the input-independent adversarial action set as
\begin{align}
\Lambda_{\AD}^{\modelind}(\Nadversarial)
\triangleq
\Big\{
    \boldsymbol{g}\in
    \Lambda_{\AD}^{\modeldep}(\Nadversarial)
    \,\Big|\,
    \exists\,g
    \text{ such that }
    g_{\bu}=g,
    \ \forall\bu\in\inputset
\Big\}.
\label{eq:independent_adversarial_action_set}
\end{align}
Hence, $\Lambda_{\AD}^{\modelind}(\Nadversarial)$ represents the
special case in which the joint adversarial-noise distribution does
not depend on the realization of $\bU$. Therefore, one can verify from
\eqref{eq:dependent_adversarial_action_set} and
\eqref{eq:independent_adversarial_action_set} that
\begin{align}
    \Lambda_{\AD}^{\modelind}(\Nadversarial)
    \subseteq
    \Lambda_{\AD}^{\modeldep}(\Nadversarial).
    \label{eq:independent_subset_dependent_action_set}
\end{align}

Upon receiving the reports, let us define the report tuple as
$\underline{\bY}\triangleq(\bY_1,\ldots,\bY_{\Nnodes})$. Consider a
deterministic realization
$\underline{\by}=(\by_1,\ldots,\by_{\Nnodes})$, where, for every
$i\in\allnodesset$, we write
$\by_i=([\by_i]_1,\ldots,[\by_i]_{\ndim})\in\mathbb{R}^{\ndim}$, and
$[\by_i]_r$ denotes the $r$-th coordinate of $\by_i$. The DC processes
the received reports in two stages.

\begin{enumerate}
    \item \textbf{Acceptance:} For every coordinate $r\in[\ndim]$, let us define the
    coordinatewise lower and upper values as
    \begin{align}
        L_r(\underline{\by})
        \triangleq
        \min_{i\in\allnodesset}[\by_i]_r, \quad R_r(\underline{\by})
        \triangleq
        \max_{i\in\allnodesset}[\by_i]_r.
        \label{eq:coordinatewise_upper_value}
    \end{align}
    We define the corresponding coordinatewise lower and upper vectors
    as
    \begin{align}
        \bL(\underline{\by})
        \triangleq
        \left(
            L_1(\underline{\by}),
            \ldots,
            L_{\ndim}(\underline{\by})
        \right),
        \label{eq:coordinatewise_lower_vector}
    \end{align}
    and
    \begin{align}
        \bR(\underline{\by})
        \triangleq
        \left(
            R_1(\underline{\by}),
            \ldots,
            R_{\ndim}(\underline{\by})
        \right).
        \label{eq:coordinatewise_upper_vector}
    \end{align}

    For a threshold parameter $\eta$, we define the acceptance event as
    \begin{align}
        \acceptanceevent{\eta}
        \triangleq
        \left\{
            \left\|
                \bR(\underline{\bY})
                -
                \bL(\underline{\bY})
            \right\|_2
            \leq
            \eta\Delta
        \right\}.
        \label{eq:coordinatewise_acceptance_event}
    \end{align}
    For an adversarial strategy
    $\boldsymbol{g}\in
    \Lambda_{\AD}^{\modeldep}(\Nadversarial)$, we define the
    probability of acceptance as
    \begin{align}
        \pa(\boldsymbol{g},\eta)
        \triangleq
        \Pr\!\left(
            \acceptanceevent{\eta}
        \right).
        \label{eq:probability_of_acceptance}
    \end{align}
    The probability in \eqref{eq:probability_of_acceptance} is
    evaluated with respect to the joint randomness of the
    ground-truth vector $\bU$, the honest noises
    $\{\bN_k\}_{k\in\honestset}$, and the adversarial noises
    $\{\bN_q\}_{q\in\adversaryset}$, where, conditional on
    $\bU=\bu$, the joint adversarial-noise vector is generated
    according to $g_{\bu}(\cdot)$ defined in
    \eqref{eq:conditional_adversarial_pdf}.

    \item \textbf{Estimation:} If the reports are accepted, the DC employs the fixed
    coordinatewise midrange estimator
    \begin{align}
        \widehat{\bU}
        \triangleq
        \est(\underline{\bY})
        \triangleq
        \frac{
            \bR(\underline{\bY})
            +
            \bL(\underline{\bY})
        }{2}.
        \label{eq:coordinatewise_midrange_estimator}
    \end{align}
    For an adversarial strategy
    $\boldsymbol{g}\in
    \Lambda_{\AD}^{\modeldep}(\Nadversarial)$, we define the accepted
    mean-squared estimation error as
    \begin{align}
        \mse(\boldsymbol{g},\eta)
        \triangleq
        \mE
        \left[
            \left\|
                \bU-\widehat{\bU}
            \right\|_2^2
            \,\middle|\,
            \acceptanceevent{\eta}
        \right].
        \label{eq:accepted_mse}
    \end{align}
    The expectation in \eqref{eq:accepted_mse} is evaluated with
    respect to the joint randomness of $\bU$, the honest noises
    $\{\bN_k\}_{k\in\honestset}$, and the adversarial noises
    $\{\bN_q\}_{q\in\adversaryset}$. 
\end{enumerate}

The relationship between
\eqref{eq:coordinatewise_acceptance_event}--%
\eqref{eq:coordinatewise_midrange_estimator}
and the acceptance and estimation rules used in the previous
game-of-coding papers is discussed in Appendix \ref{app:alignment_previous_goc}.
In particular, the proposed rule reduces to the scalar max--min
and midrange rules when $\ndim=1$
\cite{GoCJournal,GoDSybil}, and to the two-node vector distance and
midpoint rules when $\Nnodes=2$ \cite{nodehi2026game}.

We next define the action set of the DC. A natural lower endpoint for
the threshold is $\eta=2$, since the distance between two honest
reports can be as large as $2\Delta$ under
\eqref{eq:honest_noise_bound}. Accordingly, we define the DC's action
set as
\begin{align}
    \Lambda_{\DC}
    \triangleq
    [2,\infty).
    \label{eq:dc_action_set}
\end{align}
Note that, the threshold $\eta$ controls the tradeoff between liveness and
estimation accuracy. Increasing $\eta$ generally increases the
probability of acceptance, but also gives the adversary greater
flexibility to distort the estimate. Decreasing $\eta$ restricts the
possible estimation error, but increases the risk that the submitted
reports are rejected. Acceptance is also valuable to the adversary in
incentive-oriented decentralized systems, since the adversary can
receive a reward or influence the final estimate only when the reports
are accepted \cite{nodehi2026gameBits}.

We model the interaction between the DC and the adversary as a
Stackelberg game \cite{von2010market}. The DC acts as the leader and
first selects $\eta\in\Lambda_{\DC}$. After observing $\eta$, the
adversary acts as the follower and selects
$\boldsymbol{g}\in
\Lambda_{\AD}^{\modeldep}(\Nadversarial)$. Once the realization
$\bU=\bu$ becomes available to the adversary, it generates its joint
noise vector according to $g_{\bu}$.

We define the utilities of the DC and the adversary as
\begin{align}
    \sU_{\DC}(\boldsymbol{g},\eta)
    &\triangleq
    Q_{\DC}
    \left(
        \mse(\boldsymbol{g},\eta),
        \pa(\boldsymbol{g},\eta)
    \right),
    \label{eq:dc_utility}\\
    \sU_{\AD}(\boldsymbol{g},\eta)
    &\triangleq
    Q_{\AD}
    \left(
        \mse(\boldsymbol{g},\eta),
        \pa(\boldsymbol{g},\eta)
    \right).
    \label{eq:adversary_utility}
\end{align}
The utility functions $Q_{\DC}$ and $Q_{\AD}$ are publicly known. We
assume that $Q_{\DC}$ is non-increasing in its first argument and
non-decreasing in its second argument. We assume that $Q_{\AD}$ is
strictly increasing in both arguments.

For $\sigma\in\{\modeldep,\modelind\}$, let
$\gameinstance{\sigma}$ denote the Stackelberg game defined above when
the adversary selects its strategy from
$\Lambda_{\AD}^{\sigma}(\Nadversarial)$. Hence,
$\gameinstance{\modeldep}$ denotes the input-dependent game and
$\gameinstance{\modelind}$ denotes the corresponding
input-independent game. All remaining system parameters and utility
functions are identical in these two game instances.

For a fixed threshold $\eta\in\Lambda_{\DC}$, we define the
adversary's best-response set in $\gameinstance{\sigma}$ as
\begin{align}
    \mathcal{B}_{\AD,\sigma}^{\eta}
    \triangleq
    \argmax_{
        \boldsymbol{g}\in
        \Lambda_{\AD}^{\sigma}(\Nadversarial)
    }
    \sU_{\AD}(\boldsymbol{g},\eta).
    \label{eq:adversarial_best_response}
\end{align}
Different strategies in
$\mathcal{B}_{\AD,\sigma}^{\eta}$ may provide different utilities to
the DC. We therefore define the set of worst-case adversarial best
responses as
\begin{align}
    \overline{\mathcal{B}}_{\AD,\sigma}^{\eta}
    \triangleq
    \argmin_{
        \boldsymbol{g}\in
        \mathcal{B}_{\AD,\sigma}^{\eta}
    }
    \sU_{\DC}(\boldsymbol{g},\eta).
    \label{eq:worst_case_adversarial_response}
\end{align}
The optimal threshold selected by the DC is defined by
\begin{align}
    \eta_{\sigma}^{*}
    \in
    \argmax_{\eta\in\Lambda_{\DC}}
    \;
    \min_{
        \boldsymbol{g}\in
        \mathcal{B}_{\AD,\sigma}^{\eta}
    }
    \sU_{\DC}(\boldsymbol{g},\eta).
    \label{eq:optimal_threshold}
\end{align}
For any
$\boldsymbol{g}_{\sigma}^{*}\in
\overline{\mathcal{B}}_{\AD,\sigma}^{\eta_{\sigma}^{*}}$, we call
$(\eta_{\sigma}^{*},\boldsymbol{g}_{\sigma}^{*})$ a Stackelberg
equilibrium of $\gameinstance{\sigma}$.

\section{Main Results}\label{sec:main_results}

The main objective of this section is to determine whether allowing the
joint adversarial-noise distribution to depend on the realization of
$\bU$ provides any additional strategic advantage to the adversary.
We show that every input-dependent adversarial strategy can be
replaced by an input-independent adversarial strategy that induces
exactly the same probability of acceptance and accepted mean-squared
estimation error. Consequently, the two models have the same
achievable performance regions, best-response performance pairs,
optimal thresholds, and equilibrium utilities.

Consider an arbitrary input-dependent adversarial strategy
$\boldsymbol{g}\in
\Lambda_{\AD}^{\modeldep}(\Nadversarial)$. For every
$\underline{\bz}_{\adversaryset}
\in(\mathbb{R}^{\ndim})^{\Nadversarial}$, we define the averaged joint
adversarial-noise PDF as
\begin{align}
    \overline{g}
    \left(
        \underline{\bz}_{\adversaryset}
    \right)
    \triangleq
    \mE_{\bU}
    \left[
        g_{\bU}
        \left(
            \underline{\bz}_{\adversaryset}
        \right)
    \right],
    \label{eq:averaged_adversarial_pdf}
\end{align}
where the expectation is evaluated only with respect to the
distribution of $\bU$. Thus,
$\overline{g}(\underline{\bz}_{\adversaryset})$ is obtained by
averaging the conditional PDFs
$g_{\bu}(\underline{\bz}_{\adversaryset})$ over all possible
realizations of $\bU$. For a continuous ground-truth vector, this
expectation is the usual integral weighted by its PDF, whereas for a
discrete ground-truth vector, it is the corresponding weighted sum.

Using the averaged PDF $\overline{g}$, we define an
input-independent adversarial strategy
$\overline{\boldsymbol{g}}$ by requiring the adversary to use the same
joint noise PDF $\overline{g}$ for every realization of $\bU$. More
precisely, we define
\begin{align}
    \overline{\boldsymbol{g}}
    \triangleq
    \left\{
        \overline{g}_{\bu}
    \right\}_{\bu\in\inputset},
    \qquad
    \overline{g}_{\bu}
    =
    \overline{g},
    \quad
    \forall\bu\in\inputset.
    \label{eq:averaged_input_independent_strategy}
\end{align}
Thus, unlike the original strategy $\boldsymbol{g}$, the strategy
$\overline{\boldsymbol{g}}$ does not change with the realization of
$\bU$.

\begin{theorem}
\label{thm:exact_averaged_strategy_reduction}

For every
$\boldsymbol{g}\in
\Lambda_{\AD}^{\modeldep}(\Nadversarial)$, the function
$\overline{g}$ defined in
\eqref{eq:averaged_adversarial_pdf} is a valid joint adversarial-noise
PDF, and the strategy
$\overline{\boldsymbol{g}}$ defined in
\eqref{eq:averaged_input_independent_strategy} satisfies
\begin{align}
    \overline{\boldsymbol{g}}
    \in
    \Lambda_{\AD}^{\modelind}(\Nadversarial).
    \label{eq:averaged_strategy_is_independent}
\end{align}
Moreover, for every $\eta\in\Lambda_{\DC}$, we have
\begin{align}
    \pa(\boldsymbol{g},\eta)
    =
    \pa(\overline{\boldsymbol{g}},\eta),
    \label{eq:averaged_strategy_same_pa}
\end{align}
and
\begin{align}
    \mse(\boldsymbol{g},\eta)
    =
    \mse(\overline{\boldsymbol{g}},\eta).
    \label{eq:averaged_strategy_same_mse}
\end{align}

\end{theorem}

Theorem~\ref{thm:exact_averaged_strategy_reduction} states that the
dependence of the adversarial-noise distribution on $\bU$ can be
removed without changing either performance metric. The proof relies
on the facts that the acceptance event and the error of the fixed
coordinatewise midrange estimator depend only on the honest and
adversarial noises. The proof is provided in
Appendix~\ref{app:proof_averaged_strategy_reduction}.

Since the utilities in \eqref{eq:dc_utility} and
\eqref{eq:adversary_utility} depend on the adversarial strategy only
through $\mse$ and $\pa$, Theorem~\ref{thm:exact_averaged_strategy_reduction}
immediately yields the following result.

\begin{corollary}[Utility preservation]
\label{cor:averaged_strategy_utility_preservation}

For every
$\boldsymbol{g}\in
\Lambda_{\AD}^{\modeldep}(\Nadversarial)$ and every
$\eta\in\Lambda_{\DC}$, the corresponding averaged strategy satisfies
\begin{align}
    \sU_{\DC}(\boldsymbol{g},\eta)
    &=
    \sU_{\DC}(\overline{\boldsymbol{g}},\eta),
    \label{eq:averaged_strategy_same_dc_utility}\\
    \sU_{\AD}(\boldsymbol{g},\eta)
    &=
    \sU_{\AD}(\overline{\boldsymbol{g}},\eta).
    \label{eq:averaged_strategy_same_ad_utility}
\end{align}

\end{corollary}

We next formalize the implication of
Theorem~\ref{thm:exact_averaged_strategy_reduction} for the complete
set of achievable performance pairs. For
$\sigma\in\{\modeldep,\modelind\}$ and every
$\eta\in\Lambda_{\DC}$, we define
\begin{align}
\performanceregion{\sigma}(\eta)
\triangleq
\Big\{
    \big(
        \pa(\boldsymbol{g},\eta),
        \mse(\boldsymbol{g},\eta)
    \big)
    \,\Big|\,
    \boldsymbol{g}\in
    \Lambda_{\AD}^{\sigma}(\Nadversarial)
\Big\}.
\label{eq:achievable_performance_region}
\end{align}

\begin{theorem}
\label{thm:performance_region_equivalence}

For every $\eta\in\Lambda_{\DC}$, the input-dependent and
input-independent adversarial models induce the same achievable
performance region. More precisely,
\begin{align}
    \performanceregion{\modeldep}(\eta)
    =
    \performanceregion{\modelind}(\eta),
    \qquad
    \forall\eta\in\Lambda_{\DC}.
    \label{eq:main_performance_region_equality}
\end{align}

\end{theorem}

The inclusion
$\performanceregion{\modelind}(\eta)
\subseteq
\performanceregion{\modeldep}(\eta)$
follows from
\eqref{eq:independent_subset_dependent_action_set}. The reverse
inclusion follows by applying
Theorem~\ref{thm:exact_averaged_strategy_reduction} to every
input-dependent adversarial strategy. A complete proof is provided in
Appendix~\ref{app:proof_performance_region_equivalence}.

We finally state the game-theoretic consequence of
Theorem~\ref{thm:performance_region_equivalence}. Since the two
adversarial models induce the same achievable performance region for
every $\eta\in\Lambda_{\DC}$, they induce the same adversarial
best-response utility and the same worst-case utility for the DC at
every threshold.

\begin{theorem}
\label{thm:stackelberg_equivalence}

The input-dependent and input-independent games have the same set of
optimal thresholds for the DC. More precisely,
\begin{align}
&\argmax_{\eta\in\Lambda_{\DC}}
\;
\min_{
    \boldsymbol{g}\in
    \mathcal{B}_{\AD,\modeldep}^{\eta}
}
\sU_{\DC}(\boldsymbol{g},\eta)
=
\argmax_{\eta\in\Lambda_{\DC}}
\;
\min_{
    \boldsymbol{g}\in
    \mathcal{B}_{\AD,\modelind}^{\eta}
}
\sU_{\DC}(\boldsymbol{g},\eta).
\label{eq:optimal_threshold_set_equality}
\end{align}
Consider an arbitrary threshold $\eta^{*}$ belonging to the common set
in \eqref{eq:optimal_threshold_set_equality}, and let
\begin{align}
    \boldsymbol{g}_{\modeldep}^{*}
    &\in
    \overline{\mathcal{B}}_{\AD,\modeldep}^{\eta^{*}},
    \qquad
    \boldsymbol{g}_{\modelind}^{*}
    \in
    \overline{\mathcal{B}}_{\AD,\modelind}^{\eta^{*}}
\label{eq:equilibrium_adversarial_strategies}
\end{align}
be arbitrary corresponding worst-case adversarial best responses.
Then, the equilibrium utilities of the DC and the adversary are
identical in the two games. That is,
\begin{align}
    \sU_{\DC}
    \left(
        \boldsymbol{g}_{\modeldep}^{*},
        \eta^{*}
    \right)
    &=
    \sU_{\DC}
    \left(
        \boldsymbol{g}_{\modelind}^{*},
        \eta^{*}
    \right),
\label{eq:equilibrium_dc_utility_equality}\\
    \sU_{\AD}
    \left(
        \boldsymbol{g}_{\modeldep}^{*},
        \eta^{*}
    \right)
    &=
    \sU_{\AD}
    \left(
        \boldsymbol{g}_{\modelind}^{*},
        \eta^{*}
    \right).
\label{eq:equilibrium_ad_utility_equality}
\end{align}

\end{theorem}

Theorem~\ref{thm:stackelberg_equivalence} shows that allowing the
adversarial-noise distribution to depend on the realization of $\bU$
does not change the set of thresholds that are optimal for the DC.
Moreover, for any threshold in this common set and any corresponding
worst-case adversarial best responses, both players obtain the same
equilibrium utilities in the two games. The proof is provided in
Appendix~\ref{app:proof_stackelberg_equivalence}.

\appendices

\section{Relationship with Previous Game-of-Coding Models}
\label{app:alignment_previous_goc}

In this appendix, we show that the acceptance and estimation mechanism
defined in
\eqref{eq:coordinatewise_upper_value}--%
\eqref{eq:coordinatewise_midrange_estimator}
provides a common generalization of the mechanisms considered in the
previous game-of-coding papers
\cite{GoCJournal,GoDSybil,nodehi2026game}.
In particular, the scalar multi-node model of \cite{GoDSybil} and the
two-node vector model of \cite{nodehi2026game} are obtained as exact
special cases of the present formulation.

The previous game-of-coding formulations assume that the adversarial
noise distribution is independent of the ground-truth computation.
Under the notation of the present paper, this corresponds to
restricting the conditional adversarial PDFs in
\eqref{eq:conditional_adversarial_pdf} according to
\begin{align}
    g_{\bu}
    \left(
        \underline{\bz}_{\adversaryset}
    \right)
    =
    g
    \left(
        \underline{\bz}_{\adversaryset}
    \right),
    \qquad
    \forall \bu\in\inputset.
    \label{eq:app_input_independent_specialization}
\end{align}
Equivalently, the adversarial strategy is restricted to
$\Lambda_{\AD}^{\modelind}(\Nadversarial)$, which is a subset of
$\Lambda_{\AD}^{\modeldep}(\Nadversarial)$ by
\eqref{eq:independent_subset_dependent_action_set}. Consequently, once
we establish that allowing the adversarial noise distribution to
depend on $\bU$ does not change the achievable probability of
acceptance and accepted mean-squared error in the present model, the
same conclusion follows immediately for the previous 
game-of-coding models after applying the corresponding specializations
below.

\subsection{Scalar Game of Coding}
\label{app:scalar_previous_models}

We first specialize the present formulation to scalar-valued
computations by setting $\ndim=1$. We identify every one-dimensional
vector with its unique scalar coordinate and write
\begin{align}
    \bU=(U),
    \qquad
    \bY_i=(Y_i),
    \qquad
    \bN_i=(N_i),
    \qquad
    i\in\allnodesset.
    \label{eq:app_scalar_identification}
\end{align}
It follows from the definition of the coordinatewise extrema in
\eqref{eq:coordinatewise_upper_value} that
\begin{align}
    L_1(\underline{\bY})
    &=
    \min_{i\in\allnodesset}Y_i,
    \label{eq:app_scalar_lower_value}\\
    R_1(\underline{\bY})
    &=
    \max_{i\in\allnodesset}Y_i.
    \label{eq:app_scalar_upper_value}
\end{align}
Therefore, the acceptance statistic in
\eqref{eq:coordinatewise_acceptance_event} satisfies
\begin{align}
&\left\|
    \bR(\underline{\bY})
    -
    \bL(\underline{\bY})
\right\|_2
\nonumber\\
&\qquad
\overset{(a)}{=}
\left|
    R_1(\underline{\bY})
    -
    L_1(\underline{\bY})
\right|
\overset{(b)}{=}
R_1(\underline{\bY})
-
L_1(\underline{\bY})
\nonumber\\
&\qquad
\overset{(c)}{=}
\max_{i\in\allnodesset}Y_i
-
\min_{i\in\allnodesset}Y_i,
\label{eq:app_scalar_acceptance_statistic}
\end{align}
where $(a)$ follows because the Euclidean norm of a
one-dimensional vector is the absolute value of its unique
coordinate, $(b)$ follows from
$R_1(\underline{\bY})\geq L_1(\underline{\bY})$, and $(c)$ follows
from \eqref{eq:app_scalar_lower_value} and
\eqref{eq:app_scalar_upper_value}. Consequently, the acceptance event
reduces to
\begin{align}
    \acceptanceevent{\eta}
    =
    \left\{
        \max_{i\in\allnodesset}Y_i
        -
        \min_{i\in\allnodesset}Y_i
        \leq
        \eta\Delta
    \right\}.
    \label{eq:app_scalar_acceptance_event}
\end{align}
Equation \eqref{eq:app_scalar_acceptance_event} is exactly the
max--min acceptance rule employed in the scalar multi-node
Sybil-resistant game of coding in \cite{GoDSybil}.

Similarly, the estimator in
\eqref{eq:coordinatewise_midrange_estimator} satisfies
\begin{align}
    \widehat U
    &\overset{(a)}{=}
    \left[
        \widehat{\bU}
    \right]_1
    \overset{(b)}{=}
    \frac{
        R_1(\underline{\bY})
        +
        L_1(\underline{\bY})
    }{2}
    \nonumber\\
    &\overset{(c)}{=}
    \frac{
        \max_{i\in\allnodesset}Y_i
        +
        \min_{i\in\allnodesset}Y_i
    }{2},
    \label{eq:app_scalar_midrange_estimator}
\end{align}
where $(a)$ follows from the scalar identification in
\eqref{eq:app_scalar_identification}, $(b)$ follows from
\eqref{eq:coordinatewise_midrange_estimator}, and $(c)$ follows from
\eqref{eq:app_scalar_lower_value} and
\eqref{eq:app_scalar_upper_value}. Thus, the fixed estimator of the
present paper reduces exactly to the scalar midrange estimator used in
\cite{GoDSybil}.

\subsection{Two-Node Vector-Valued Game of Coding}
\label{app:vector_previous_model}

We now specialize the present formulation to $\Nnodes=2$, while
allowing an arbitrary dimension $\ndim\geq1$. Consider a deterministic
pair of reports
$\underline{\by}=(\by_1,\by_2)$, where
$\by_1,\by_2\in\mathbb{R}^{\ndim}$. For every coordinate
$r\in[\ndim]$, the definitions in
\eqref{eq:coordinatewise_upper_value} yield
\begin{align}
&R_r(\underline{\by})-L_r(\underline{\by})
\nonumber\\
&\qquad
\overset{(a)}{=}
\max\left\{
    [\by_1]_r,[\by_2]_r
\right\}
-
\min\left\{
    [\by_1]_r,[\by_2]_r
\right\}
\nonumber\\
&\qquad
\overset{(b)}{=}
\left|
    [\by_1]_r-[\by_2]_r
\right|,
\label{eq:app_vector_coordinatewise_difference}
\end{align}
where $(a)$ follows from the definitions of the coordinatewise lower
and upper values, and $(b)$ follows from the identity
$\max\{x_1,x_2\}-\min\{x_1,x_2\}=|x_1-x_2|$ for any
$x_1,x_2\in\mathbb{R}$.

Using \eqref{eq:app_vector_coordinatewise_difference}, we obtain
\begin{align}
&\left\|
    \bR(\underline{\by})-\bL(\underline{\by})
\right\|_2^2
\nonumber\\
&\qquad
\overset{(a)}{=}
\sum_{r=1}^{\ndim}
\left(
    R_r(\underline{\by})-L_r(\underline{\by})
\right)^2
\nonumber\\
&\qquad
\overset{(b)}{=}
\sum_{r=1}^{\ndim}
\left|
    [\by_1]_r-[\by_2]_r
\right|^2
\nonumber\\
&\qquad
\overset{(c)}{=}
\left\|
    \by_1-\by_2
\right\|_2^2,
\label{eq:app_vector_distance_squared}
\end{align}
where $(a)$ follows from
\eqref{eq:coordinatewise_lower_vector} and
\eqref{eq:coordinatewise_upper_vector}, $(b)$ follows from
\eqref{eq:app_vector_coordinatewise_difference}, and $(c)$ follows
from the definition of the Euclidean norm. Taking the square root of
both sides of \eqref{eq:app_vector_distance_squared} gives
\begin{align}
    \left\|
        \bR(\underline{\by})-\bL(\underline{\by})
    \right\|_2
    =
    \left\|
        \by_1-\by_2
    \right\|_2.
\label{eq:app_vector_distance}
\end{align}
Therefore, when $\Nnodes=2$, the acceptance event in
\eqref{eq:coordinatewise_acceptance_event} reduces to
\begin{align}
    \acceptanceevent{\eta}
    =
    \left\{
        \left\|
            \bY_1-\bY_2
        \right\|_2
        \leq
        \eta\Delta
    \right\},
\label{eq:app_vector_two_node_acceptance}
\end{align}
which is exactly the acceptance rule employed in the two-node
vector-valued game of coding in \cite{nodehi2026game}.

We next consider the estimation rule. For every
$r\in[\ndim]$, we have
\begin{align}
&\frac{
    R_r(\underline{\by})+L_r(\underline{\by})
}{2}
\nonumber\\
&\qquad
\overset{(a)}{=}
\frac{
    \max\left\{
        [\by_1]_r,[\by_2]_r
    \right\}
    +
    \min\left\{
        [\by_1]_r,[\by_2]_r
    \right\}
}{2}
\nonumber\\
&\qquad
\overset{(b)}{=}
\frac{
    [\by_1]_r+[\by_2]_r
}{2},
\label{eq:app_vector_coordinatewise_midpoint}
\end{align}
where $(a)$ follows from the definitions of the coordinatewise lower
and upper values, and $(b)$ follows from the identity
$\max\{x_1,x_2\}+\min\{x_1,x_2\}=x_1+x_2$.

Since \eqref{eq:app_vector_coordinatewise_midpoint} holds for every
coordinate $r\in[\ndim]$, it follows from
\eqref{eq:coordinatewise_lower_vector} and
\eqref{eq:coordinatewise_upper_vector} that
\begin{align}
    \frac{
        \bR(\underline{\by})+\bL(\underline{\by})
    }{2}
    =
    \frac{
        \by_1+\by_2
    }{2}.
\label{eq:app_vector_midpoint}
\end{align}
Consequently, when $\Nnodes=2$, the estimator in
\eqref{eq:coordinatewise_midrange_estimator} reduces to
\begin{align}
    \widehat{\bU}
    =
    \frac{
        \bY_1+\bY_2
    }{2},
\label{eq:app_vector_two_node_estimator}
\end{align}
which is exactly the midpoint estimator employed in
\cite{nodehi2026game}. Hence, the acceptance and estimation rules of
the two-node vector-valued game of coding are exact special cases of
the present formulation.

\section{Proof of Theorem~\ref{thm:exact_averaged_strategy_reduction}}
\label{app:proof_averaged_strategy_reduction}

In this appendix, we prove
Theorem~\ref{thm:exact_averaged_strategy_reduction}. The central idea
of the proof is that, although the conditional adversarial-noise PDF
$g_{\bu}$ may depend on the realization $\bU=\bu$, the acceptance
event and the estimation error depend only on the realized honest and
adversarial noises. After averaging $g_{\bu}$ over the distribution of
$\bU$, the resulting input-independent strategy induces exactly the
same joint distribution of all noise variables.

The proof proceeds in three steps. In
Lemma~\ref{lem:averaged_pdf_is_valid}, we show that the averaged
function $\overline{g}$ defined in
\eqref{eq:averaged_adversarial_pdf} is a valid joint PDF and induces an
input-independent adversarial strategy. In
Lemma~\ref{lem:joint_noise_distribution_equivalence}, we show that the
input-dependent strategy $\boldsymbol{g}$ and its averaged
input-independent counterpart
$\overline{\boldsymbol{g}}$ induce the same joint distribution of the
honest and adversarial noises. In
Lemma~\ref{lem:noise_domain_representation}, we show that both the
acceptance event and the estimation error can be written entirely in
terms of these noises. We then combine these results to prove
Theorem~\ref{thm:exact_averaged_strategy_reduction}.

\begin{lemma}
\label{lem:averaged_pdf_is_valid}

For every
$\boldsymbol{g}\in
\Lambda_{\AD}^{\modeldep}(\Nadversarial)$, the function
$\overline{g}$ defined in
\eqref{eq:averaged_adversarial_pdf} is a valid joint PDF on
$\big(\mathbb{R}^{\ndim}\big)^{\Nadversarial}$. Moreover, the strategy
$\overline{\boldsymbol{g}}$ defined in
\eqref{eq:averaged_input_independent_strategy} satisfies
\begin{align}
    \overline{\boldsymbol{g}}
    \in
    \Lambda_{\AD}^{\modelind}(\Nadversarial).
    \label{eq:app_averaged_strategy_membership}
\end{align}

\end{lemma}

\begin{proof}

Since
$\boldsymbol{g}\in
\Lambda_{\AD}^{\modeldep}(\Nadversarial)$, for every
$\bu\in\inputset$ and every
$\underline{\bz}_{\adversaryset}
\in
\big(\mathbb{R}^{\ndim}\big)^{\Nadversarial}$, we have
\begin{align}
    g_{\bu}
    \left(
        \underline{\bz}_{\adversaryset}
    \right)
    \geq
    0.
    \label{eq:app_conditional_pdf_nonnegative}
\end{align}
Therefore, for every
$\underline{\bz}_{\adversaryset}
\in
\big(\mathbb{R}^{\ndim}\big)^{\Nadversarial}$,
\begin{align}
    \overline{g}
    \left(
        \underline{\bz}_{\adversaryset}
    \right)
    &\overset{(a)}{=}
    \mE_{\bU}
    \left[
        g_{\bU}
        \left(
            \underline{\bz}_{\adversaryset}
        \right)
    \right]
    \overset{(b)}{\geq}
    0,
    \label{eq:app_averaged_pdf_nonnegative}
\end{align}
where $(a)$ follows from
\eqref{eq:averaged_adversarial_pdf}, and $(b)$ follows from
\eqref{eq:app_conditional_pdf_nonnegative}.

We next verify that $\overline{g}$ integrates to one. Since
$g_{\bu}$ is a valid joint PDF for every $\bu\in\inputset$, we have
\begin{align}
    \int_{\left(\mathbb{R}^{\ndim}\right)^{\Nadversarial}}
    g_{\bu}
    \left(
        \underline{\bz}_{\adversaryset}
    \right)
    d\underline{\bz}_{\adversaryset}
    =
    1,
    \qquad
    \forall\bu\in\inputset.
    \label{eq:app_conditional_pdf_normalization}
\end{align}
Consequently,
\begin{align}
&\int_{\left(\mathbb{R}^{\ndim}\right)^{\Nadversarial}}
\overline{g}
\left(
    \underline{\bz}_{\adversaryset}
\right)
d\underline{\bz}_{\adversaryset}
\nonumber\\
&\qquad
\overset{(a)}{=}
\int_{\left(\mathbb{R}^{\ndim}\right)^{\Nadversarial}}
\mE_{\bU}
\left[
    g_{\bU}
    \left(
        \underline{\bz}_{\adversaryset}
    \right)
\right]
d\underline{\bz}_{\adversaryset}
\nonumber\\
&\qquad
\overset{(b)}{=}
\mE_{\bU}
\left[
    \int_{\left(\mathbb{R}^{\ndim}\right)^{\Nadversarial}}
    g_{\bU}
    \left(
        \underline{\bz}_{\adversaryset}
    \right)
    d\underline{\bz}_{\adversaryset}
\right]
\nonumber\\
&\qquad
\overset{(c)}{=}
\mE_{\bU}[1]
\overset{(d)}{=}
1,
\label{eq:app_averaged_pdf_normalization}
\end{align}
where $(a)$ follows from
\eqref{eq:averaged_adversarial_pdf}, $(b)$ follows by exchanging the
expectation and the integral, which is valid because the integrand is
nonnegative, $(c)$ follows from
\eqref{eq:app_conditional_pdf_normalization}, and $(d)$ follows because
the expectation of the constant one is one. Equations \eqref{eq:app_averaged_pdf_nonnegative} and
\eqref{eq:app_averaged_pdf_normalization} show that
$\overline{g}$ is a valid joint PDF on
$\big(\mathbb{R}^{\ndim}\big)^{\Nadversarial}$.

By
\eqref{eq:averaged_input_independent_strategy}, the strategy
$\overline{\boldsymbol{g}}$ uses the same joint PDF
$\overline{g}$ for every realization $\bu\in\inputset$. Therefore,
\begin{align}
    \overline{g}_{\bu}
    =
    \overline{g},
    \qquad
    \forall\bu\in\inputset.
    \label{eq:app_averaged_strategy_does_not_depend_on_u}
\end{align}
It follows from
\eqref{eq:independent_adversarial_action_set},
\eqref{eq:app_averaged_pdf_normalization}, and
\eqref{eq:app_averaged_strategy_does_not_depend_on_u} that
\begin{align}
    \overline{\boldsymbol{g}}
    \in
    \Lambda_{\AD}^{\modelind}(\Nadversarial).
\end{align}
This proves the lemma.

\end{proof}

\begin{lemma}
\label{lem:joint_noise_distribution_equivalence}

Let
$\boldsymbol{g}\in
\Lambda_{\AD}^{\modeldep}(\Nadversarial)$ and let
$\overline{\boldsymbol{g}}$ be the corresponding averaged
input-independent strategy defined in
\eqref{eq:averaged_input_independent_strategy}. Then the two strategies
induce the same marginal joint PDF for the honest- and
adversarial-noise vectors. More precisely,
\begin{align}
&f_{\underline{\bN}_{\honestset},
\underline{\bN}_{\adversaryset}}^{\boldsymbol{g}}
\left(
    \underline{\bx}_{\honestset},
    \underline{\bz}_{\adversaryset}
\right)
\nonumber\\
&\qquad =
f_{\underline{\bN}_{\honestset},
\underline{\bN}_{\adversaryset}}^{\overline{\boldsymbol{g}}}
\left(
    \underline{\bx}_{\honestset},
    \underline{\bz}_{\adversaryset}
\right)
\nonumber\\
&\qquad =
f_{\underline{\bN}_{\honestset}}
\left(
    \underline{\bx}_{\honestset}
\right)
\overline{g}
\left(
    \underline{\bz}_{\adversaryset}
\right).
\label{eq:app_joint_noise_pdf_equivalence}
\end{align}

\end{lemma}

\begin{proof}

Since the honest noises are independent and identically distributed
according to $f_{\bN}$, their joint PDF is
\begin{align}
    f_{\underline{\bN}_{\honestset}}
    \left(
        \underline{\bx}_{\honestset}
    \right)
    =
    \prod_{k\in\honestset}
    f_{\bN}(\bx_k).
    \label{eq:app_joint_honest_noise_pdf}
\end{align}

We first derive the marginal joint PDF of the honest and adversarial
noises under the input-dependent strategy $\boldsymbol{g}$. By the
law of total probability,
\begin{align}
&f_{\underline{\bN}_{\honestset},
\underline{\bN}_{\adversaryset}}^{\boldsymbol{g}}
\left(
    \underline{\bx}_{\honestset},
    \underline{\bz}_{\adversaryset}
\right)
\nonumber\\
&\qquad
\overset{(a)}{=}
\mE_{\bU}
\left[
f_{\underline{\bN}_{\honestset},
\underline{\bN}_{\adversaryset}\mid\bU}^{\boldsymbol{g}}
\left(
    \underline{\bx}_{\honestset},
    \underline{\bz}_{\adversaryset}
    \mid\bU
\right)
\right]
\nonumber\\
&\qquad
\overset{(b)}{=}
\mE_{\bU}
\left[
f_{\underline{\bN}_{\honestset}}
\left(
    \underline{\bx}_{\honestset}
\right)
g_{\bU}
\left(
    \underline{\bz}_{\adversaryset}
\right)
\right]
\nonumber\\
&\qquad
\overset{(c)}{=}
f_{\underline{\bN}_{\honestset}}
\left(
    \underline{\bx}_{\honestset}
\right)
\mE_{\bU}
\left[
g_{\bU}
\left(
    \underline{\bz}_{\adversaryset}
\right)
\right]
\nonumber\\
&\qquad
\overset{(d)}{=}
f_{\underline{\bN}_{\honestset}}
\left(
    \underline{\bx}_{\honestset}
\right)
\overline{g}
\left(
    \underline{\bz}_{\adversaryset}
\right),
\label{eq:app_dependent_joint_noise_pdf}
\end{align}
where $(a)$ follows from marginalizing over the randomness of $\bU$,
$(b)$ follows from
\eqref{eq:conditional_honest_adversarial_factorization}, $(c)$ follows
because
$f_{\underline{\bN}_{\honestset}}
(\underline{\bx}_{\honestset})$
does not depend on $\bU$, and $(d)$ follows from
\eqref{eq:averaged_adversarial_pdf}.

We next derive the corresponding marginal joint PDF under
$\overline{\boldsymbol{g}}$. By
\eqref{eq:averaged_input_independent_strategy}, the conditional
adversarial-noise PDF is $\overline{g}$ for every realization of
$\bU$. Therefore,
\begin{align}
&f_{\underline{\bN}_{\honestset},
\underline{\bN}_{\adversaryset}}^{\overline{\boldsymbol{g}}}
\left(
    \underline{\bx}_{\honestset},
    \underline{\bz}_{\adversaryset}
\right)
\nonumber\\
&\qquad
\overset{(a)}{=}
\mE_{\bU}
\left[
f_{\underline{\bN}_{\honestset},
\underline{\bN}_{\adversaryset}\mid\bU}
^{\overline{\boldsymbol{g}}}
\left(
    \underline{\bx}_{\honestset},
    \underline{\bz}_{\adversaryset}
    \mid\bU
\right)
\right]
\nonumber\\
&\qquad
\overset{(b)}{=}
\mE_{\bU}
\left[
f_{\underline{\bN}_{\honestset}}
\left(
    \underline{\bx}_{\honestset}
\right)
\overline{g}
\left(
    \underline{\bz}_{\adversaryset}
\right)
\right]
\nonumber\\
&\qquad
\overset{(c)}{=}
f_{\underline{\bN}_{\honestset}}
\left(
    \underline{\bx}_{\honestset}
\right)
\overline{g}
\left(
    \underline{\bz}_{\adversaryset}
\right),
\label{eq:app_independent_joint_noise_pdf}
\end{align}
where $(a)$ follows from marginalizing over $\bU$, $(b)$ follows from
the conditional independence assumption together with
\eqref{eq:averaged_input_independent_strategy}, and $(c)$ follows
because both factors inside the expectation are independent of
$\bU$. Combining \eqref{eq:app_dependent_joint_noise_pdf} and
\eqref{eq:app_independent_joint_noise_pdf} gives
\eqref{eq:app_joint_noise_pdf_equivalence}. This proves the lemma.

\end{proof}

\begin{lemma}
\label{lem:noise_domain_representation}

Let
\begin{align}
    \underline{\bN}
    \triangleq
    \left(
        \bN_1,\ldots,\bN_{\Nnodes}
    \right)
    \label{eq:app_complete_noise_tuple}
\end{align}
denote the complete tuple of honest and adversarial noises. The
acceptance event and the estimation error satisfy
\begin{align}
    \acceptanceevent{\eta}
    &=
    \left\{
        \left\|
            \bR(\underline{\bN})
            -
            \bL(\underline{\bN})
        \right\|_2
        \leq
        \eta\Delta
    \right\},
    \label{eq:app_acceptance_noise_domain}\\
    \left\|
        \bU-\widehat{\bU}
    \right\|_2^2
    &=
    \left\|
        \frac{
            \bR(\underline{\bN})
            +
            \bL(\underline{\bN})
        }{2}
    \right\|_2^2.
    \label{eq:app_error_noise_domain}
\end{align}
Therefore, both quantities depend on the noises but not directly on
the realization of $\bU$.

\end{lemma}

\begin{proof}

For every node $i\in\allnodesset$ and every coordinate
$r\in[\ndim]$, it follows from
\eqref{eq:honest_report} and \eqref{eq:adversarial_report} that
\begin{align}
    [\bY_i]_r
    =
    [\bU]_r+[\bN_i]_r.
    \label{eq:app_report_coordinate}
\end{align}
Using \eqref{eq:app_report_coordinate}, the coordinatewise lower value
satisfies
\begin{align}
    L_r(\underline{\bY})
    &\overset{(a)}{=}
    \min_{i\in\allnodesset}
    \left\{
        [\bU]_r+[\bN_i]_r
    \right\}
    \nonumber\\
    &\overset{(b)}{=}
    [\bU]_r
    +
    \min_{i\in\allnodesset}
    [\bN_i]_r
    \nonumber\\
    &\overset{(c)}{=}
    [\bU]_r
    +
    L_r(\underline{\bN}),
    \label{eq:app_lower_translation}
\end{align}
where $(a)$ follows from
\eqref{eq:coordinatewise_upper_value} and
\eqref{eq:app_report_coordinate}, $(b)$ follows because
$[\bU]_r$ is common to all reports, and $(c)$ follows from the
definition of the coordinatewise lower value applied to the noise
tuple.

Similarly, the coordinatewise upper value satisfies
\begin{align}
    R_r(\underline{\bY})
    &\overset{(a)}{=}
    \max_{i\in\allnodesset}
    \left\{
        [\bU]_r+[\bN_i]_r
    \right\}
    \nonumber\\
    &\overset{(b)}{=}
    [\bU]_r
    +
    \max_{i\in\allnodesset}
    [\bN_i]_r
    \nonumber\\
    &\overset{(c)}{=}
    [\bU]_r
    +
    R_r(\underline{\bN}),
    \label{eq:app_upper_translation}
\end{align}
where $(a)$ follows from
\eqref{eq:coordinatewise_upper_value} and
\eqref{eq:app_report_coordinate}, $(b)$ follows because
$[\bU]_r$ is common to all reports, and $(c)$ follows from the
definition of the coordinatewise upper value applied to the noise
tuple.

Since \eqref{eq:app_lower_translation} and
\eqref{eq:app_upper_translation} hold for every
$r\in[\ndim]$, we obtain
\begin{align}
    \bL(\underline{\bY})
    &=
    \bU+\bL(\underline{\bN}),
    \label{eq:app_lower_vector_translation}\\
    \bR(\underline{\bY})
    &=
    \bU+\bR(\underline{\bN}).
    \label{eq:app_upper_vector_translation}
\end{align}
It follows that
\begin{align}
&\bR(\underline{\bY})
-
\bL(\underline{\bY})
\nonumber\\
&\qquad
\overset{(a)}{=}
\bU+\bR(\underline{\bN})
-
\left(
    \bU+\bL(\underline{\bN})
\right)
\nonumber\\
&\qquad
\overset{(b)}{=}
\bR(\underline{\bN})
-
\bL(\underline{\bN}),
\label{eq:app_range_translation_invariance}
\end{align}
where $(a)$ follows from
\eqref{eq:app_lower_vector_translation} and
\eqref{eq:app_upper_vector_translation}, and $(b)$ follows by
canceling $\bU$. Substituting
\eqref{eq:app_range_translation_invariance} into
\eqref{eq:coordinatewise_acceptance_event} gives
\eqref{eq:app_acceptance_noise_domain}.

For the estimation error, we have
\begin{align}
    \widehat{\bU}
    &\overset{(a)}{=}
    \frac{
        \bR(\underline{\bY})
        +
        \bL(\underline{\bY})
    }{2}
    \nonumber\\
    &\overset{(b)}{=}
    \frac{
        \bU+\bR(\underline{\bN})
        +
        \bU+\bL(\underline{\bN})
    }{2}
    \nonumber\\
    &\overset{(c)}{=}
    \bU
    +
    \frac{
        \bR(\underline{\bN})
        +
        \bL(\underline{\bN})
    }{2},
    \label{eq:app_estimator_noise_decomposition}
\end{align}
where $(a)$ follows from
\eqref{eq:coordinatewise_midrange_estimator}, $(b)$ follows from
\eqref{eq:app_lower_vector_translation} and
\eqref{eq:app_upper_vector_translation}, and $(c)$ follows by
collecting the two copies of $\bU$. Therefore,
\begin{align}
    \left\|
        \bU-\widehat{\bU}
    \right\|_2^2
    &\overset{(a)}{=}
    \left\|
        -
        \frac{
            \bR(\underline{\bN})
            +
            \bL(\underline{\bN})
        }{2}
    \right\|_2^2
    \nonumber\\
    &\overset{(b)}{=}
    \left\|
        \frac{
            \bR(\underline{\bN})
            +
            \bL(\underline{\bN})
        }{2}
    \right\|_2^2,
    \label{eq:app_squared_error_noise_only}
\end{align}
where $(a)$ follows from
\eqref{eq:app_estimator_noise_decomposition}, and $(b)$ follows from
$\|-\bx\|_2^2=\|\bx\|_2^2$. Equation
\eqref{eq:app_squared_error_noise_only} proves
\eqref{eq:app_error_noise_domain}. This completes the proof.

\end{proof}

We are now ready to prove
Theorem~\ref{thm:exact_averaged_strategy_reduction}.

\begin{proof}[Proof of
Theorem~\ref{thm:exact_averaged_strategy_reduction}]

Consider an arbitrary input-dependent adversarial strategy
\begin{align}
    \boldsymbol{g}
    \in
    \Lambda_{\AD}^{\modeldep}(\Nadversarial).
    \label{eq:app_arbitrary_dependent_strategy}
\end{align}
Let $\overline{g}$ and
$\overline{\boldsymbol{g}}$ be defined according to
\eqref{eq:averaged_adversarial_pdf} and
\eqref{eq:averaged_input_independent_strategy}, respectively.
Lemma~\ref{lem:averaged_pdf_is_valid} establishes that
$\overline{g}$ is a valid joint PDF and that
\begin{align}
    \overline{\boldsymbol{g}}
    \in
    \Lambda_{\AD}^{\modelind}(\Nadversarial).
    \label{eq:app_theorem_independent_membership}
\end{align}

We next prove equality of the probabilities of acceptance. For
deterministic realizations
$\underline{\bx}_{\honestset}$ and
$\underline{\bz}_{\adversaryset}$, let
$\underline{\bn}
    \left(
        \underline{\bx}_{\honestset},
        \underline{\bz}_{\adversaryset}
    \right)$
denote the complete node-indexed noise tuple whose honest components
are given by $\underline{\bx}_{\honestset}$ and whose adversarial
components are given by
$\underline{\bz}_{\adversaryset}$. We define the noise-domain
acceptance region as
\begin{align}
\mathcal{D}_{\eta}
\triangleq
\Bigg\{
\left(
    \underline{\bx}_{\honestset},
    \underline{\bz}_{\adversaryset}
\right)
\ \Bigg|\
\left\|
\bR
\left(
    \underline{\bn}
    \left(
        \underline{\bx}_{\honestset},
        \underline{\bz}_{\adversaryset}
    \right)
\right)
-
\bL
\left(
    \underline{\bn}
    \left(
        \underline{\bx}_{\honestset},
        \underline{\bz}_{\adversaryset}
    \right)
\right)
\right\|_2
\leq
\eta\Delta
\Bigg\}.
\label{eq:app_noise_acceptance_region}
\end{align}
By Lemma~\ref{lem:noise_domain_representation}, the reports are
accepted if and only if the realized honest and adversarial noises
belong to $\mathcal{D}_{\eta}$.
Under the input-dependent strategy $\boldsymbol{g}$, the probability
of acceptance satisfies
\begin{align}
\pa(\boldsymbol{g},\eta)
&
\overset{(a)}{=}
\int_{\left(\mathbb{R}^{\ndim}\right)^{\Nhonest}}
\int_{\left(\mathbb{R}^{\ndim}\right)^{\Nadversarial}}
\mathbbm{1}
\left\{
\left(
    \underline{\bx}_{\honestset},
    \underline{\bz}_{\adversaryset}
\right)
\in
\mathcal{D}_{\eta}
\right\}
\nonumber\\
&\hspace{40mm}\times
f_{\underline{\bN}_{\honestset},
\underline{\bN}_{\adversaryset}}^{\boldsymbol{g}}
\left(
    \underline{\bx}_{\honestset},
    \underline{\bz}_{\adversaryset}
\right)
d\underline{\bz}_{\adversaryset}
d\underline{\bx}_{\honestset}
\nonumber\\
&\quad
\overset{(b)}{=}
\int_{\left(\mathbb{R}^{\ndim}\right)^{\Nhonest}}
\int_{\left(\mathbb{R}^{\ndim}\right)^{\Nadversarial}}
\mathbbm{1}
\left\{
\left(
    \underline{\bx}_{\honestset},
    \underline{\bz}_{\adversaryset}
\right)
\in
\mathcal{D}_{\eta}
\right\}
\nonumber\\
&\hspace{40mm}\times
f_{\underline{\bN}_{\honestset}}
\left(
    \underline{\bx}_{\honestset}
\right)
\overline{g}
\left(
    \underline{\bz}_{\adversaryset}
\right)
d\underline{\bz}_{\adversaryset}
d\underline{\bx}_{\honestset}
\nonumber\\
&\quad
\overset{(c)}{=}
\int_{\left(\mathbb{R}^{\ndim}\right)^{\Nhonest}}
\int_{\left(\mathbb{R}^{\ndim}\right)^{\Nadversarial}}
\mathbbm{1}
\left\{
\left(
    \underline{\bx}_{\honestset},
    \underline{\bz}_{\adversaryset}
\right)
\in
\mathcal{D}_{\eta}
\right\}
\nonumber\\
&\hspace{40mm}\times
f_{\underline{\bN}_{\honestset},
\underline{\bN}_{\adversaryset}}^{\overline{\boldsymbol{g}}}
\left(
    \underline{\bx}_{\honestset},
    \underline{\bz}_{\adversaryset}
\right)
d\underline{\bz}_{\adversaryset}
d\underline{\bx}_{\honestset}
\nonumber\\
&\quad
\overset{(d)}{=}
\pa(\overline{\boldsymbol{g}},\eta),
\label{eq:app_probability_acceptance_equality}
\end{align}
where $(a)$ follows from
\eqref{eq:probability_of_acceptance},
Lemma~\ref{lem:noise_domain_representation}, and
\eqref{eq:app_noise_acceptance_region}, $(b)$ follows from
\eqref{eq:app_dependent_joint_noise_pdf}, $(c)$ follows from
\eqref{eq:app_independent_joint_noise_pdf}, and $(d)$ follows from the
definition of the probability of acceptance under
$\overline{\boldsymbol{g}}$. This proves
\begin{align}
    \pa(\boldsymbol{g},\eta)
    =
    \pa(\overline{\boldsymbol{g}},\eta).
\end{align}

It remains to prove equality of the accepted mean-squared estimation
errors. For every deterministic pair
$(\underline{\bx}_{\honestset},
\underline{\bz}_{\adversaryset})$, define
\begin{align}
e
\left(
    \underline{\bx}_{\honestset},
    \underline{\bz}_{\adversaryset}
\right)
\triangleq
\left\|
\frac{
\bR
\left(
    \underline{\bn}
    \left(
        \underline{\bx}_{\honestset},
        \underline{\bz}_{\adversaryset}
    \right)
\right)
+
\bL
\left(
    \underline{\bn}
    \left(
        \underline{\bx}_{\honestset},
        \underline{\bz}_{\adversaryset}
    \right)
\right)
}{2}
\right\|_2^2.
\label{eq:app_deterministic_noise_error}
\end{align}
By Lemma~\ref{lem:noise_domain_representation}, this function is
exactly the squared estimation error corresponding to the given noise
realization.
Define the unnormalized accepted squared error under
$\boldsymbol{g}$ as
\begin{align}
    J(\boldsymbol{g},\eta)
    \triangleq
    \mE
    \left[
        \left\|
            \bU-\widehat{\bU}
        \right\|_2^2
        \mathbbm{1}
        \left\{
            \acceptanceevent{\eta}
        \right\}
    \right].
    \label{eq:app_unnormalized_accepted_error}
\end{align}
Using \eqref{eq:app_noise_acceptance_region} and
\eqref{eq:app_deterministic_noise_error}, we obtain
\begin{align}
&J(\boldsymbol{g},\eta)
\nonumber\\
&\quad
\overset{(a)}{=}
\int_{\left(\mathbb{R}^{\ndim}\right)^{\Nhonest}}
\int_{\left(\mathbb{R}^{\ndim}\right)^{\Nadversarial}}
e
\left(
    \underline{\bx}_{\honestset},
    \underline{\bz}_{\adversaryset}
\right)
\mathbbm{1}
\left\{
\left(
    \underline{\bx}_{\honestset},
    \underline{\bz}_{\adversaryset}
\right)
\in
\mathcal{D}_{\eta}
\right\}
\nonumber\\
&\hspace{40mm}\times
f_{\underline{\bN}_{\honestset},
\underline{\bN}_{\adversaryset}}^{\boldsymbol{g}}
\left(
    \underline{\bx}_{\honestset},
    \underline{\bz}_{\adversaryset}
\right)
d\underline{\bz}_{\adversaryset}
d\underline{\bx}_{\honestset}
\nonumber\\
&\quad
\overset{(b)}{=}
\int_{\left(\mathbb{R}^{\ndim}\right)^{\Nhonest}}
\int_{\left(\mathbb{R}^{\ndim}\right)^{\Nadversarial}}
e
\left(
    \underline{\bx}_{\honestset},
    \underline{\bz}_{\adversaryset}
\right)
\mathbbm{1}
\left\{
\left(
    \underline{\bx}_{\honestset},
    \underline{\bz}_{\adversaryset}
\right)
\in
\mathcal{D}_{\eta}
\right\}
\nonumber\\
&\hspace{40mm}\times
f_{\underline{\bN}_{\honestset}}
\left(
    \underline{\bx}_{\honestset}
\right)
\overline{g}
\left(
    \underline{\bz}_{\adversaryset}
\right)
d\underline{\bz}_{\adversaryset}
d\underline{\bx}_{\honestset}
\nonumber\\
&\quad
\overset{(c)}{=}
\int_{\left(\mathbb{R}^{\ndim}\right)^{\Nhonest}}
\int_{\left(\mathbb{R}^{\ndim}\right)^{\Nadversarial}}
e
\left(
    \underline{\bx}_{\honestset},
    \underline{\bz}_{\adversaryset}
\right)
\mathbbm{1}
\left\{
\left(
    \underline{\bx}_{\honestset},
    \underline{\bz}_{\adversaryset}
\right)
\in
\mathcal{D}_{\eta}
\right\}
\nonumber\\
&\hspace{40mm}\times
f_{\underline{\bN}_{\honestset},
\underline{\bN}_{\adversaryset}}^{\overline{\boldsymbol{g}}}
\left(
    \underline{\bx}_{\honestset},
    \underline{\bz}_{\adversaryset}
\right)
d\underline{\bz}_{\adversaryset}
d\underline{\bx}_{\honestset}
\nonumber\\
&\quad
\overset{(d)}{=}
J(\overline{\boldsymbol{g}},\eta),
\label{eq:app_unnormalized_error_equality}
\end{align}
where $(a)$ follows from
\eqref{eq:app_unnormalized_accepted_error},
Lemma~\ref{lem:noise_domain_representation},
\eqref{eq:app_noise_acceptance_region}, and
\eqref{eq:app_deterministic_noise_error}, $(b)$ follows from
\eqref{eq:app_dependent_joint_noise_pdf}, $(c)$ follows from
\eqref{eq:app_independent_joint_noise_pdf}, and $(d)$ follows from the
definition of
$J(\overline{\boldsymbol{g}},\eta)$.

Suppose that the common probability of acceptance is positive, and
define
\begin{align}
    p_{\eta}
    \triangleq
    \pa(\boldsymbol{g},\eta)
    \overset{(a)}{=}
    \pa(\overline{\boldsymbol{g}},\eta)
    >
    0,
    \label{eq:app_common_positive_acceptance}
\end{align}
where $(a)$ follows from
\eqref{eq:app_probability_acceptance_equality}. The accepted
mean-squared estimation error under $\boldsymbol{g}$ satisfies
\begin{align}
    \mse(\boldsymbol{g},\eta)
    &\overset{(a)}{=}
    \frac{
        J(\boldsymbol{g},\eta)
    }{
        \pa(\boldsymbol{g},\eta)
    }
    \nonumber\\
    &\overset{(b)}{=}
    \frac{
        J(\overline{\boldsymbol{g}},\eta)
    }{
        \pa(\overline{\boldsymbol{g}},\eta)
    }
    \nonumber\\
    &\overset{(c)}{=}
    \mse(\overline{\boldsymbol{g}},\eta),
    \label{eq:app_mse_equality}
\end{align}
where $(a)$ follows from the definition of conditional expectation
and \eqref{eq:app_common_positive_acceptance}, $(b)$ follows from
\eqref{eq:app_probability_acceptance_equality} and
\eqref{eq:app_unnormalized_error_equality}, and $(c)$ follows from the
definition of
$\mse(\overline{\boldsymbol{g}},\eta)$.

We have shown that
$\overline{\boldsymbol{g}}$ is an input-independent adversarial
strategy and that, for every $\eta\in\Lambda_{\DC}$,
\begin{align}
    \pa(\boldsymbol{g},\eta)
    &=
    \pa(\overline{\boldsymbol{g}},\eta),\\
    \mse(\boldsymbol{g},\eta)
    &=
    \mse(\overline{\boldsymbol{g}},\eta)
\end{align}
whenever the common probability of acceptance is positive. This
completes the proof of
Theorem~\ref{thm:exact_averaged_strategy_reduction}.

\end{proof}

\section{Proof of Theorem~\ref{thm:performance_region_equivalence}}
\label{app:proof_performance_region_equivalence}
In this appendix, we prove
Theorem~\ref{thm:performance_region_equivalence}. Recall from
\eqref{eq:independent_subset_dependent_action_set} that every
input-independent adversarial strategy is also an admissible
input-dependent adversarial strategy. Therefore, for every
$\eta\in\Lambda_{\DC}$,
\begin{align}
    \performanceregion{\modelind}(\eta)
    \subseteq
    \performanceregion{\modeldep}(\eta).
    \label{eq:app_independent_region_subset_intro}
\end{align}
Thus, to prove the equality in
\eqref{eq:main_performance_region_equality}, it remains only to
establish the reverse inclusion
\begin{align}
    \performanceregion{\modeldep}(\eta)
    \subseteq
    \performanceregion{\modelind}(\eta),
    \qquad
    \forall\eta\in\Lambda_{\DC}.
    \label{eq:app_dependent_region_subset_intro}
\end{align}
The main idea is to consider an arbitrary
performance pair achieved by an input-dependent strategy and apply
Theorem~\ref{thm:exact_averaged_strategy_reduction}. The theorem
constructs an input-independent strategy that induces exactly the same
probability of acceptance and accepted mean-squared estimation error.
Hence, every performance pair achievable in the input-dependent model
is also achievable in the input-independent model.

\begin{proof}[Proof of
Theorem~\ref{thm:performance_region_equivalence}]

Fix an arbitrary threshold $\eta\in\Lambda_{\DC}$. We first prove
\begin{align}
    \performanceregion{\modelind}(\eta)
    \subseteq
    \performanceregion{\modeldep}(\eta).
    \label{eq:app_independent_region_subset}
\end{align}
Consider an arbitrary pair
$(p,e)\in\performanceregion{\modelind}(\eta)$. By
\eqref{eq:achievable_performance_region}, there exists
$\boldsymbol{g}_{\modelind}\in
\Lambda_{\AD}^{\modelind}(\Nadversarial)$ such that
$p=\pa(\boldsymbol{g}_{\modelind},\eta)$,
$e=\mse(\boldsymbol{g}_{\modelind},\eta)$. Moreover,
\eqref{eq:independent_subset_dependent_action_set} implies that
$\boldsymbol{g}_{\modelind}\in
\Lambda_{\AD}^{\modeldep}(\Nadversarial)$. Therefore,
\begin{align}
    (p,e)
    &\overset{(a)}{=}
    \left(
        \pa(\boldsymbol{g}_{\modelind},\eta),
        \mse(\boldsymbol{g}_{\modelind},\eta)
    \right)
    \overset{(b)}{\in}
    \performanceregion{\modeldep}(\eta),
    \label{eq:app_independent_pair_in_dependent_region}
\end{align}
where $(a)$ follows from the choice of
$\boldsymbol{g}_{\modelind}$, and $(b)$ follows from
\eqref{eq:achievable_performance_region},
\eqref{eq:independent_subset_dependent_action_set}. Since $(p,e)$ was arbitrary,
\eqref{eq:app_independent_region_subset} follows. We next prove the reverse inclusion
\begin{align}
    \performanceregion{\modeldep}(\eta)
    \subseteq
    \performanceregion{\modelind}(\eta).
    \label{eq:app_dependent_region_subset}
\end{align}
Consider an arbitrary pair
$(p,e)\in\performanceregion{\modeldep}(\eta)$. By
\eqref{eq:achievable_performance_region}, there exists
$\boldsymbol{g}_{\modeldep}\in
\Lambda_{\AD}^{\modeldep}(\Nadversarial)$ such that
$p=\pa(\boldsymbol{g}_{\modeldep},\eta)$,
$e=\mse(\boldsymbol{g}_{\modeldep},\eta)$.

Let $\overline{\boldsymbol{g}}_{\modeldep}$ denote the averaged
input-independent strategy constructed from
$\boldsymbol{g}_{\modeldep}$ according to
\eqref{eq:averaged_adversarial_pdf} and
\eqref{eq:averaged_input_independent_strategy}. By
Theorem~\ref{thm:exact_averaged_strategy_reduction},
$\overline{\boldsymbol{g}}_{\modeldep}\in
\Lambda_{\AD}^{\modelind}(\Nadversarial)$ and
\begin{align}
    \pa(\boldsymbol{g}_{\modeldep},\eta)
    &=
    \pa(\overline{\boldsymbol{g}}_{\modeldep},\eta),
    \label{eq:app_region_proof_pa_preservation}\\
    \mse(\boldsymbol{g}_{\modeldep},\eta)
    &=
    \mse(\overline{\boldsymbol{g}}_{\modeldep},\eta).
    \label{eq:app_region_proof_mse_preservation}
\end{align}
Therefore, 
\begin{align}
    (p,e)
    &\overset{(a)}{=}
    \left(
        \pa(\boldsymbol{g}_{\modeldep},\eta),
        \mse(\boldsymbol{g}_{\modeldep},\eta)
    \right)
    \nonumber\\
    &\overset{(b)}{=}
    \left(
        \pa(\overline{\boldsymbol{g}}_{\modeldep},\eta),
        \mse(\overline{\boldsymbol{g}}_{\modeldep},\eta)
    \right)
    \overset{(c)}{\in}
    \performanceregion{\modelind}(\eta),
    \label{eq:app_dependent_pair_in_independent_region}
\end{align}
where $(a)$ follows from the choice of
$\boldsymbol{g}_{\modeldep}$, $(b)$ follows from
\eqref{eq:app_region_proof_pa_preservation} and
\eqref{eq:app_region_proof_mse_preservation}, and $(c)$ follows from
\eqref{eq:achievable_performance_region} and the fact that
$\overline{\boldsymbol{g}}_{\modeldep}\in
\Lambda_{\AD}^{\modelind}(\Nadversarial)$. 

Since $(p,e)$ was
arbitrary, \eqref{eq:app_dependent_region_subset} follows. Combining \eqref{eq:app_independent_region_subset} and
\eqref{eq:app_dependent_region_subset}, we obtain
\begin{align}
    \performanceregion{\modeldep}(\eta)
    =
    \performanceregion{\modelind}(\eta).
    \label{eq:app_performance_region_equality}
\end{align}
Since $\eta\in\Lambda_{\DC}$ was arbitrary,
\eqref{eq:app_performance_region_equality} holds for every
$\eta\in\Lambda_{\DC}$. This proves
\eqref{eq:main_performance_region_equality} and completes the proof of
Theorem~\ref{thm:performance_region_equivalence}.

\end{proof}

\section{Proof of Theorem~\ref{thm:stackelberg_equivalence}}
\label{app:proof_stackelberg_equivalence}

In this appendix, we prove
Theorem~\ref{thm:stackelberg_equivalence}. The main idea is that
Theorem~\ref{thm:performance_region_equivalence} establishes that, for
every threshold $\eta\in\Lambda_{\DC}$, the input-dependent and
input-independent adversarial models induce exactly the same set of
achievable probability-of-acceptance and mean-squared-error pairs.
Since the utilities of both players depend on an adversarial strategy
only through these two quantities, the adversary obtains the same
optimal utility in the two models for every fixed threshold.
Furthermore, the DC faces the same worst-case utility among the
adversary's best responses. Therefore, the two games induce the same
optimization problem for the DC and, consequently, the same set of
optimal thresholds.

\begin{proof}[Proof of
Theorem~\ref{thm:stackelberg_equivalence}]

Fix an arbitrary threshold $\eta\in\Lambda_{\DC}$. By
Theorem~\ref{thm:performance_region_equivalence}, we have
\begin{align}
    \performanceregion{\modeldep}(\eta)
    =
    \performanceregion{\modelind}(\eta).
    \label{eq:app_stackelberg_common_region}
\end{align}
We first show that the adversary obtains the same optimal utility in
the two games at the fixed threshold $\eta$. Using the definition of
the achievable performance region in
\eqref{eq:achievable_performance_region}, we have
\begin{align}
&\max_{
    \boldsymbol{g}\in
    \Lambda_{\AD}^{\modeldep}(\Nadversarial)
}
\sU_{\AD}(\boldsymbol{g},\eta)
\nonumber\\
&\qquad
\overset{(a)}{=}
\max_{
    (p,e)\in
    \performanceregion{\modeldep}(\eta)
}
Q_{\AD}(e,p)
\nonumber\\
&\qquad
\overset{(b)}{=}
\max_{
    (p,e)\in
    \performanceregion{\modelind}(\eta)
}
Q_{\AD}(e,p)
\nonumber\\
&\qquad
\overset{(c)}{=}
\max_{
    \boldsymbol{g}\in
    \Lambda_{\AD}^{\modelind}(\Nadversarial)
}
\sU_{\AD}(\boldsymbol{g},\eta),
\label{eq:app_adversary_optimal_utility_fixed_eta}
\end{align}
where $(a)$ follows from the definition of
$\sU_{\AD}$ in \eqref{eq:adversary_utility} and the definition of
$\performanceregion{\modeldep}(\eta)$, $(b)$ follows from
\eqref{eq:app_stackelberg_common_region}, and $(c)$ follows from the
definition of $\sU_{\AD}$ and the definition of
$\performanceregion{\modelind}(\eta)$.
Let us define the set of performance pairs that maximize the
adversary's utility at threshold $\eta$ as
\begin{align}
    \mathcal{P}_{\AD}^{\eta}
    \triangleq
    \argmax_{
        (p,e)\in
        \performanceregion{\modeldep}(\eta)
    }
    Q_{\AD}(e,p).
    \label{eq:app_common_adversarial_best_performance_set}
\end{align}
Because of \eqref{eq:app_stackelberg_common_region}, the same set can
equivalently be written as
\begin{align}
    \mathcal{P}_{\AD}^{\eta}
    =
    \argmax_{
        (p,e)\in
        \performanceregion{\modelind}(\eta)
    }
    Q_{\AD}(e,p).
    \label{eq:app_common_adversarial_best_performance_set_ind}
\end{align}

We next characterize the performance pairs induced by the
adversary's best responses. More precisely, we prove that, for every
$\sigma\in\{\modeldep,\modelind\}$, the set of performance pairs
induced by the strategies in
$\mathcal{B}_{\AD,\sigma}^{\eta}$ (defined in \eqref{eq:adversarial_best_response}) is exactly
$\mathcal{P}_{\AD}^{\eta}$. That is,
\begin{align}
\mathcal{P}_{\AD}^{\eta}
=
\Big\{
    \big(
        \pa(\boldsymbol{g},\eta),
        \mse(\boldsymbol{g},\eta)
    \big)
    \,\Big|\,
    \boldsymbol{g}\in
    \mathcal{B}_{\AD,\sigma}^{\eta}
\Big\}.
\label{eq:app_best_response_performance_characterization}
\end{align}

To verify the inclusion from right to left in
\eqref{eq:app_best_response_performance_characterization}, consider
an arbitrary
$\boldsymbol{g}\in\mathcal{B}_{\AD,\sigma}^{\eta}$ and define
\begin{align}
    p_{\boldsymbol{g}}
    \triangleq
    \pa(\boldsymbol{g},\eta),
    \qquad
    e_{\boldsymbol{g}}
    \triangleq
    \mse(\boldsymbol{g},\eta).
    \label{eq:app_best_response_induced_pair}
\end{align}
For every
$(p,e)\in\performanceregion{\sigma}(\eta)$, there exists
$\widetilde{\boldsymbol{g}}\in
\Lambda_{\AD}^{\sigma}(\Nadversarial)$ that induces $(p,e)$. Hence,
\begin{align}
    Q_{\AD}
    \left(
        e_{\boldsymbol{g}},
        p_{\boldsymbol{g}}
    \right)
    &\overset{(a)}{=}
    \sU_{\AD}(\boldsymbol{g},\eta)
    \nonumber\\
    &\overset{(b)}{\geq}
    \sU_{\AD}(\widetilde{\boldsymbol{g}},\eta)
    \nonumber\\
    &\overset{(c)}{=}
    Q_{\AD}(e,p),
    \label{eq:app_best_response_pair_maximizes_utility}
\end{align}
where $(a)$ and $(c)$ follow from
\eqref{eq:adversary_utility}, and $(b)$ follows from the definition of
$\mathcal{B}_{\AD,\sigma}^{\eta}$ in
\eqref{eq:adversarial_best_response}. Since \eqref{eq:app_best_response_pair_maximizes_utility} holds for
every $(p,e)\in\performanceregion{\sigma}(\eta)$, the pair
$(p_{\boldsymbol{g}},e_{\boldsymbol{g}})$ maximizes
$Q_{\AD}(e,p)$ over $\performanceregion{\sigma}(\eta)$. Therefore, by
the definition of $\mathcal{P}_{\AD}^{\eta}$ in
\eqref{eq:app_common_adversarial_best_performance_set} and
\eqref{eq:app_common_adversarial_best_performance_set_ind}, we have
\begin{align}
    (p_{\boldsymbol{g}},e_{\boldsymbol{g}})
    \in
    \mathcal{P}_{\AD}^{\eta}.
\end{align}

For the reverse inclusion, consider an arbitrary
$(p,e)\in\mathcal{P}_{\AD}^{\eta}$. By the definition of the
achievable performance region, there exists
$\boldsymbol{g}\in
\Lambda_{\AD}^{\sigma}(\Nadversarial)$ satisfying
\begin{align}
    p=\pa(\boldsymbol{g},\eta),
    \qquad
    e=\mse(\boldsymbol{g},\eta).
    \label{eq:app_strategy_realizing_optimal_pair}
\end{align}
For any
$\widetilde{\boldsymbol{g}}\in
\Lambda_{\AD}^{\sigma}(\Nadversarial)$, let
$\widetilde p=\pa(\widetilde{\boldsymbol{g}},\eta)$ and
$\widetilde e=\mse(\widetilde{\boldsymbol{g}},\eta)$. Since
$(p,e)\in\mathcal{P}_{\AD}^{\eta}$, we obtain
\begin{align}
    \sU_{\AD}(\boldsymbol{g},\eta)
    &\overset{(a)}{=}
    Q_{\AD}(e,p)
    \nonumber\\
    &\overset{(b)}{\geq}
    Q_{\AD}(\widetilde e,\widetilde p)
    \nonumber\\
    &\overset{(c)}{=}
    \sU_{\AD}(\widetilde{\boldsymbol{g}},\eta),
    \label{eq:app_optimal_pair_induces_best_response}
\end{align}
where $(a)$ follows from
\eqref{eq:app_strategy_realizing_optimal_pair} and
\eqref{eq:adversary_utility}, $(b)$ follows from the definition of
$\mathcal{P}_{\AD}^{\eta}$ in
\eqref{eq:app_common_adversarial_best_performance_set} and
\eqref{eq:app_common_adversarial_best_performance_set_ind}, and $(c)$
follows from \eqref{eq:adversary_utility}. Since
\eqref{eq:app_optimal_pair_induces_best_response} holds for every
$\widetilde{\boldsymbol{g}}\in
\Lambda_{\AD}^{\sigma}(\Nadversarial)$, the strategy
$\boldsymbol{g}$ maximizes the adversary's utility over
$\Lambda_{\AD}^{\sigma}(\Nadversarial)$. Therefore, by the definition
of the adversarial best-response set in
\eqref{eq:adversarial_best_response}, we have
\begin{align}
    \boldsymbol{g}
    \in
    \mathcal{B}_{\AD,\sigma}^{\eta}.
\end{align}
This proves the reverse inclusion and completes the proof of
\eqref{eq:app_best_response_performance_characterization}.

It now follows that the worst-case utility of the DC at the fixed
threshold $\eta$ is the same in the two games. In fact,
\begin{align}
&\min_{
    \boldsymbol{g}\in
    \mathcal{B}_{\AD,\modeldep}^{\eta}
}
\sU_{\DC}(\boldsymbol{g},\eta)
\nonumber\\
&\qquad
\overset{(a)}{=}
\min_{
    (p,e)\in
    \mathcal{P}_{\AD}^{\eta}
}
Q_{\DC}(e,p)
\nonumber\\
&\qquad
\overset{(b)}{=}
\min_{
    \boldsymbol{g}\in
    \mathcal{B}_{\AD,\modelind}^{\eta}
}
\sU_{\DC}(\boldsymbol{g},\eta),
\label{eq:app_fixed_threshold_dc_objective_equality}
\end{align}
where $(a)$ and $(b)$ follow from
\eqref{eq:dc_utility} and
\eqref{eq:app_best_response_performance_characterization}.

Since \eqref{eq:app_fixed_threshold_dc_objective_equality} holds for
every $\eta\in\Lambda_{\DC}$, the two functions optimized by the DC
are identical over $\Lambda_{\DC}$. Therefore,
\begin{align}
&\argmax_{\eta\in\Lambda_{\DC}}
\;
\min_{
    \boldsymbol{g}\in
    \mathcal{B}_{\AD,\modeldep}^{\eta}
}
\sU_{\DC}(\boldsymbol{g},\eta)
\overset{(a)}{=}
\argmax_{\eta\in\Lambda_{\DC}}
\;
\min_{
    \boldsymbol{g}\in
    \mathcal{B}_{\AD,\modelind}^{\eta}
}
\sU_{\DC}(\boldsymbol{g},\eta),
\label{eq:app_optimal_threshold_set_equality}
\end{align}
where $(a)$ follows from
\eqref{eq:app_fixed_threshold_dc_objective_equality}. This proves
\eqref{eq:optimal_threshold_set_equality}.

We next prove the equality of the equilibrium utilities. Consider an
arbitrary threshold $\eta^{*}$ belonging to the common set in
\eqref{eq:optimal_threshold_set_equality}, and arbitrary strategies
\begin{align}
    \boldsymbol{g}_{\modeldep}^{*}
    &\in
    \overline{\mathcal{B}}_{\AD,\modeldep}^{\eta^{*}},
    \qquad
    \boldsymbol{g}_{\modelind}^{*}
    \in
    \overline{\mathcal{B}}_{\AD,\modelind}^{\eta^{*}}.
    \label{eq:app_arbitrary_equilibrium_strategies}
\end{align}

By the definition of the worst-case best-response sets in
\eqref{eq:worst_case_adversarial_response}, we have
\begin{align}
&\sU_{\DC}
\left(
    \boldsymbol{g}_{\modeldep}^{*},
    \eta^{*}
\right)
\nonumber\\
&\qquad
\overset{(a)}{=}
\min_{
    \boldsymbol{g}\in
    \mathcal{B}_{\AD,\modeldep}^{\eta^{*}}
}
\sU_{\DC}(\boldsymbol{g},\eta^{*})
\nonumber\\
&\qquad
\overset{(b)}{=}
\min_{
    \boldsymbol{g}\in
    \mathcal{B}_{\AD,\modelind}^{\eta^{*}}
}
\sU_{\DC}(\boldsymbol{g},\eta^{*})
\nonumber\\
&\qquad
\overset{(c)}{=}
\sU_{\DC}
\left(
    \boldsymbol{g}_{\modelind}^{*},
    \eta^{*}
\right),
\label{eq:app_equilibrium_dc_utility_equality}
\end{align}
where $(a)$ and $(c)$ follow from
\eqref{eq:worst_case_adversarial_response}, and $(b)$ follows from
\eqref{eq:app_fixed_threshold_dc_objective_equality}. This proves
\eqref{eq:equilibrium_dc_utility_equality}.

Finally, since
$\boldsymbol{g}_{\modeldep}^{*}$ and
$\boldsymbol{g}_{\modelind}^{*}$ belong to their respective
adversarial best-response sets, we have
\begin{align}
&\sU_{\AD}
\left(
    \boldsymbol{g}_{\modeldep}^{*},
    \eta^{*}
\right)
\nonumber\\
&\qquad
\overset{(a)}{=}
\max_{
    \boldsymbol{g}\in
    \Lambda_{\AD}^{\modeldep}(\Nadversarial)
}
\sU_{\AD}(\boldsymbol{g},\eta^{*})
\nonumber\\
&\qquad
\overset{(b)}{=}
\max_{
    \boldsymbol{g}\in
    \Lambda_{\AD}^{\modelind}(\Nadversarial)
}
\sU_{\AD}(\boldsymbol{g},\eta^{*})
\nonumber\\
&\qquad
\overset{(c)}{=}
\sU_{\AD}
\left(
    \boldsymbol{g}_{\modelind}^{*},
    \eta^{*}
\right),
\label{eq:app_equilibrium_ad_utility_equality}
\end{align}
where $(a)$ and $(c)$ follow from the definition of the adversarial
best-response sets in \eqref{eq:adversarial_best_response}, and $(b)$
follows from
\eqref{eq:app_adversary_optimal_utility_fixed_eta}. This proves
\eqref{eq:equilibrium_ad_utility_equality} and completes the proof of
Theorem~\ref{thm:stackelberg_equivalence}.

\end{proof}

\end{document}